\documentclass[onecolumn,11pt, draftcls]{IEEEtran}

\usepackage{amsmath,stackrel}
\usepackage{amssymb}
\usepackage{amsthm}
\usepackage{arydshln}
\usepackage{graphicx}
\usepackage{bbm}
\usepackage[bb=boondox]{mathalfa}
\usepackage{subcaption}
\usepackage{caption}

\newtheorem{theorem}{Theorem}
\newtheorem{lemma}[theorem]{Lemma}
\newtheorem{definition}[theorem]{Definition}

\newcommand{\mypar}[1]{\vspace{0.1in}\noindent{\bf #1.}}
\DeclareMathAlphabet\mathbfcal{OMS}{cmsy}{b}{n}

\title{Optimal detection and error exponents for hidden multi-state processes via random duration model approach}
\author{Dragana Bajovi\'c$^1$\footnote{$^1$ D. Bajovi\'c and D. Vukobratovi\'c are with Department of Power, Electronic and Communications Engineering, Faculty of Technical Sciences, University of Novi Sad, Novi Sad, Serbia (e-mail: \{dbajovic, dejanv@uns.ac.rs\}).
\newline$^2$ K. He, L. Stankovi\'c and V. Stankovi\'c  are with Department of Electronic and Electrical Engineering, University of Strathclyde, Glasgow, G1 1XW, UK (e-mail: \{kanghang.he, vladimir.stankovic, lina.stankovic\}@strath.ac.uk)
}, Kanghang He$^2$, Lina Stankovi\'c$^2$, Dejan Vukobratovi\'c$^1$, and Vladimir Stankovi\'c$^2$}

\begin{document}
\maketitle
%

\mypar{Abstract}
We study detection of random signals corrupted by noise that over time switch their values (states) from a finite set of possible values, where the switchings occur at unknown points in time. We model such signals by means of a random duration model that to each possible state assigns a probability mass function which controls the statistics of durations of that state occurrences. Assuming two possible signal states and Gaussian noise, we derive optimal likelihood ratio test and show that it has a computationally tractable form of a matrix product, with the number of matrices involved in the product being the number of process observations. Each matrix involved in the product is of dimension equal to the sum of durations spreads of the two states, and it can be decomposed as a product of a diagonal random matrix controlled by the process observations and a sparse constant matrix which governs transitions in the sequence of states. Using this result, we show that the Neyman-Pearson error exponent is equal to the top Lyapunov exponent for the corresponding random matrices. Using theory of large deviations, we derive a lower bound on the error exponent. Finally, we show that this bound is tight by means of numerical simulations.

\mypar{Keywords} Multi-state processes, random duration model, hypothesis testing, error exponent, large deviations principle, threshold effect, Lyapunov exponent.

\section{Introduction}
\label{sec-Introduction}

The problem of detecting a signal hidden in noise is investigated. The signal to be detected is characterised as having a constant magnitude in any one state and can transition to multiple states over time. Each occurrence of a particular state has a random duration, modelled as a discrete random variable which takes values in a finite set of integers, according to a certain probability mass function associated with that state. For each given state, duration of its occurrences over time are independent and identically distributed random variables, independent of duration of other states.

Our main motivation for studying the described model comes from non intrusive appliance load monitoring (NILM) problem, i.e., detecting one or more particular appliance states, each of unknown duration, within an aggregate power signal, as obtained from smart meters. With the large-scale roll-out of smart meters worldwide, there has been increased interest in NILM, i.e., disaggregating total household energy consumption measured by the smart meter down to appliance level using purely software tools \cite{Hart1}.  NILM can enrich energy feedback, it can support smart home automation \cite{Murray:2017}, appliance retrofit decisions, and demand response measures \cite{NILM}.

Despite significant research efforts in developing efficient NILM algorithms (see \cite{NILM}, \cite{kanghang},\cite{bochao}, \cite{Parson}, \cite{kolter} and references therein), NILM is still a  challenge, especially at low sampling rates, in the order of seconds and minutes. One obstacle is lack of standardised performance measures and appropriate theoretical bounds of detectability of appliance usage, which can help estimating performance of various algorithms. A particularly challenging problem is the detection of multi-state appliances, i.e., appliances whose power consumption switches over one appliance runtime through several different values. Examples of such appliances are a dish-washer or a washing machine, where the chosen program or setting and possibly also the appliance load (e.g., with the washing machine) determines duration  that the appliance spends in each state. The difficulty there arises from the fact that the program and the load, unknown from the perspective of NILM, are non-deterministic, i.e., vary each time the same appliance is run  resulting in difficulty in detecting  in which state the appliance is. The aggregate signal minus the appliance load is considered noise for the detection problem.

The above model is also representative of signals occurring in a range of other applications. In econometrics, examples of duration signals include marital or employment status, or in general the time an individual spends in  a certain state~\cite{UggenEconometricsIndividual00}. Further examples from econometrics are time to currency alignment or time to transactions in stock market~\cite{RusselEngleTransactionsModels05}. In communication systems theory pulse-duration modulated (PDM) signals for transmitting information encoded into the pulse duration have two possible signal states: the positive value state is a pulse  whose duration is proportional to the information symbol to be encoded, and  the zero-value state in between any two pulses. The probability distribution of the state duration is then controlled by the probability distribution on the set of information symbols to be transmitted. Further binary state examples are random telegraph signals, where the signal switches between two values in a random manner\footnote{We remark that there are other stochastic models in the literature for the random telegraph signal, e.g., the Poisson model, or the hidden Markov chain model~\cite{Hero06}\cite{HeroICASSP06}.}, and the activity pattern of a certain mobile user in a cellular communication system.

In this paper,  we are interested in deriving optimal detection tests for detecting multi-state signals with random duration structure hiding in noise. We consider binary models, where occurrences of two possible states are interleaved in time. Further, we are interested in characterizing performance of optimal detection tests measured in terms of Neyman Pearson error exponent. Works on detecting multi-state signals hidden in noise, most related to our work, include~\cite{Hero06},~\cite{Agaskar15} and~\cite{RandomWalks17}. However, in contrast to the random duration model that we propose, these references model multi-state signals in noise as hidden Markov chains. Reference~\cite{Hero06} considers random telegraph signals modelled as binary Markov chains and derives the corresponding optimal detection test in the form of a product of certain measurement defined matrices. Reference~\cite{Agaskar15} considers detection of a random walk on a graph, and derives bounds on the error exponent for the Neyman-Pearson detection test. Reference~\cite{RandomWalks17} uses the method of types to generalize the results from~\cite{Agaskar15} to non-homogeneous setting where different nodes have different signal-to-noise ratios (SNR) with respect to the walk. Furthermore, reference~\cite{RandomWalks17} proves that the derived bound on the error exponent has a convex optimization form.

\mypar{Contributions} In this paper, we show that the optimal detection test, seemingly combinatorial in nature, admits a simple, linear recursion form of a product of matrices of dimension equal to the sum of the duration spreads for the two states.  Using the preceding result,  we show that the Neyman-Pearson error exponent for this problem is given by the top Lyapunov exponent~\cite{TsitsiBlondel97} for the matrices that define the recursion. The matrices have a structure of an interleaved random diagonal and (sparse) constant component that defines transitions from one state pattern to another. Thus, we reveal that a similar structural effect as with the error exponent for hidden Markov processes occurs here as well~\cite{Hero06},\cite{RandomWalks17}. Finally, using the theory of large deviations~\cite{DemboZeitouni93}, we derive a lower bound on the error exponent and demonstrate by numerical simulations that the derived bound is very close to the true error exponent.

\mypar{Paper outline} Section~\ref{sec-Setup} states the problem setup and Section~\ref{sec-Preliminaries} gives the preliminaries. Section~\ref{sec-LLR-recursion} gives main results on the form of the optimal likelihood ratio test. Section~\ref{sec-main} provides the lower bound on the error exponent, while Section~\ref{sec-proof-main} proves this result. Finally, numerical results are given in Section~\ref{sec-NumResults} and Section~\ref{sec-Conclusion} concludes the paper.

\mypar{Notation} For an arbitrary integer $n$, ${\mathbb S}^{n-1}$ denotes the probability simplex in $\mathbb R^{n}$;  $e_1$ denotes the first canonical vector and the vector (the $n$ dimensional vector with $1$ only in the first position, and having zeros in all other positions), and ${\mathbb 1}$ the vector of all ones, where we remark that the dimension should be clear from the context; $A_0$ denotes the lower shift matrix (the $0/1$ matrix with ones only on the first subdiagonal). We denote Gaussian distribution of mean value $\mu$ and standard deviation $\sigma$ by $\mathcal N(\mu,\sigma^2)$; by $p[1,n]$  an arbitrary distribution over the first $n$ integers; by $\mathcal U[1,n]$ the uniform distribution over the first $n$ integers; $\log$ denotes the natural logarithm.

\section{Problem setup}
\label{sec-Setup}
We consider the problem of detecting a signal corrupted by noise that randomly switches from one state $m$ to another, where $m=1,2,...,M$ and in each state the signal has a certain magnitude $\mu_m$. The duration that the signal spends in a given state $m$ is modelled as a discrete random variable on a given support set $\left[1, \Delta_m\right]$, and with a certain probability mass function (pmf) defined by vector $p_m \in {\mathbb S}^{\Delta_m-1}$. In this work, we consider the case when $M=2$ and we assume that for each state $m$ we know the corresponding value of the observed signal $\mu_m$. Without loss of generality, we will assume that $\mu_2>\mu_1\geq 0$. 
For each sampling time $t=1,2,...$, let $S^t=\{S_1,...,S_t\}$ denote the sequence of states until time $t$ of the signal that we wish to detect, where for each $k=1,...,t$, $S_k\in \{1,2\}$; similarly, we denote $S^{\infty}=\{S_1,S_2,...\}$. We assume that, with probability one, the first state is $S_1\equiv 1$, and, for the purpose of analysis, we set $S_0\equiv 2$. Let $X_k$ denote the signal measurement for sample time $k$, $k=1,...,t$, and, for each $t$, collect all measurements up to time $t$ in vector $X^t = (X_1,...,X_t)$. We assume that each measurement is corrupted by a zero mean additive Gaussian noise $\mathcal N (0,\sigma^2)$, where standard deviation $\sigma>0$.

\mypar{The sequence of switching times} For the sequence of states $S_1, S_2,...$, we define the sequence of times $\{T_1,T_2,...\}$, when the signal in the sequence switches from one state to another, i.e.,
\begin{equation}\label{eq-def-switching-times}
T_{i+1} = \max\{k\geq T_i+1: S_{k} = S_{T_i+1}\},\mathrm{\;for\;}i=0,1,2,...
\end{equation}
where we set $T_0\equiv 0$. We call a phase each time window $[T_i+1,T_{i+1}]$, $i=0,1,2,\ldots$, and note that during any phase, the sequence $S^{\infty}$ stays in the same state. Since $S_1\equiv 1$, all odd-numbered intervals $[T_0+1,T_1]$, $[T_2+1,T_3]$,..., where the ordering is with respect to the order of appearance, are state $1$ phases, and all even-numbered intervals $[T_1+1,T_2]$, $[T_3+1,T_4]$,... are state $2$ phases.

\mypar{Random duration model} For $n=1,2,...$, we denote by $D_{1,n}$ the difference process
\begin{equation}\label{def-state-1-durations}
D_{1,n}=T_{2n-1}-T_{2n-2},
\end{equation}
or, in words, for each $n$, $D_{1,n}$ is the duration of the $n$-th state-$1$ phase in the sequence $S^{\infty}$. We assume that durations of state-$1$ phases are independent and identically distributed (i.i.d.), with support set of all integers in the finite interval $\left[1,\Delta_1\right]$, and with pmf given by vector $p_1=\left(p_{11},p_{12},..., p_{1\Delta_1}\right)\in{\mathbb S}^{\Delta_1-1}$. Similarly, we define
\begin{equation}\label{def-state-2-durations}
D_{2,n}=T_{2n}-T_{2n-1}
\end{equation}
to be the duration of the $n$-th state-$2$ phase in the sequence $S_1,S_2,...$, for $n=1,2,...$; we assume that the $D_{2,n}$'s are i.i.d., with support set of all integers in the interval $\left[1,\Delta_2\right]$, and pmf given by vector $p_2=\left(p_{21},p_{22},...,p_{2\Delta_2}\right) \in {\mathbb S}^{\Delta_2-1}$. We also assume that durations of state-$1$ and state-$2$ phases are mutually independent. 

\mypar{Hypothesis testing problem} Using the preceding definitions, we model the signal detection problem as the following binary hypothesis testing problem:
\begin{align}
\label{eq-hypothesis-testing-problem}
\mathcal H_0:\; &X_{k}\stackrel{\mathrm{i.i.d.}}{\sim}\mathcal N(0,\sigma^2)\\
\mathcal H_1:\; &X_{k}|S^t \stackrel{\mathrm{indep.}}{\sim}
\left\{\begin{array}[+]{ll}
\mathcal N(\mu_1,\sigma^2),\, &\mathrm{if}\, S_k=1\nonumber\\
\mathcal N(\mu_2,\sigma^2),\, &\mathrm{if}\, S_k=2
\end{array}\right.,\,\mathrm{for}\,k=1,...,t,
\end{align}
%
where $D_{1,n}\sim p_1\left(1,\Delta_1\right)$ are i.i.d., $D_{2,n}\sim p_2\left(1,\Delta_2\right)$ are i.i.d., $D_{1n}$'s and $D_{2n}$'s are independent, and $S_1\equiv 1$. We remark that the model above easily generalizes to the case when the signals $X_k$ are under both hypotheses shifted for some $\mu_0\in \mathbb R$, i.e., when, under $\mathcal H=\mathcal H_0$, $X_k\sim \mathcal N(\mu_0,\sigma^2)$ and, under $\mathcal H=\mathcal H_1$, $X_k\sim \mathcal N(\mu_{S_k}+\mu_0,\sigma^2)$; see the example of appliance detection problem later in this section. The latter hypothesis testing problem reduces to the one in~\eqref{eq-hypothesis-testing-problem} by means of the change of variables $Y_k=X_k-\mu_0$.

\mypar{Illustration: Multiphase appliance detection} Suppose that we wish to detect an event that a certain appliance in a household is switched on. We consider classes of appliances whose signature signals  exhibit a multistate (multiphase) type of behavior, such as switching from high to low signal values, where the durations of phases of the same signal level can be different across a single appliance run-time and also in different run-times of the same appliance. Examples of appliances whose signatures fall into this class are, e.g., a dishwasher and a washer-dryer. This problem can be modelled by the hypothesis testing problem~\eqref{eq-hypothesis-testing-problem} where $\mu_1$ corresponds to the appliance consumption when in low state and $\mu_2$ corresponds to the appliance consumption when in high state. In this scenario, there is an underlying baseline load which can also be modelled as a Gaussian random variable of expected value $\mu_0$ and standard deviation $\sigma^2$. Since the same baseline load is present both under $\mathcal H_0$ and $\mathcal H_1$, to cast the described appliance detection problem in the format given in~\eqref{eq-hypothesis-testing-problem}, we simply subtract the value $\mu_0$ from the observed consumption signal $X_k$.

\mypar{Likelihood ratio test and Neyman-Pearson error exponent}
We denote the probability laws corresponding to $\mathcal H_0$ and $\mathcal H_1$ by $\mathbb P_0$ and $\mathbb P_1$, respectively. Similarly, the expectations with respect to $\mathbb P_0$ and $\mathbb P_1$ are denoted by $\mathbb E_0$ and $\mathbb E_1$, respectively. The probability density functions of $X^t$ under $\mathcal H_1$ and $\mathcal H_0$ are denoted by $f_{1,t}(\cdot)$ and $f_{0,t}(\cdot)$. It will also be of interest to introduce the conditional probability density function of $ X^t$ given $S^t=s^t$ (i.e., the likelihood functions), which we denote by $f_{1,t|S^t}(\cdot|s^t)$, for any $s^t$. Finally, the likelihood ratio at time $t$ denoted by $L_t$, and at a given realization of $X^t$ is computed by $L_t( X^t)= \frac{f_{1,t}(X^t )}{f_{0,t}(X^t)}$.

It is well known that the optimal detection test (both in Neyman-Pearson and Bayes sense) for problem~\eqref{eq-hypothesis-testing-problem} is the likelihood ratio test. Conditioning on the state realizations until time $t$, $S^t=s^t$, and denoting shortly $P(s^t) = \mathbb P_1(S^t=s^t)$, we have
\begin{align}
\label{eq-LLR}
L_t(X^t)&= \sum_{s^t \in \mathcal S^t} P(s^t)\frac{f_{1,t|S^t}(X^t|s^t)}
{f_{0,t}(X^t)}\nonumber\\
&= \sum_{s^t \in \mathcal S^t} P(s^t)
\frac{\prod_{k=1}^{t}\frac{1}{\sqrt {2 \pi} \sigma} e^{- \frac{(\mu_{s_k} - X_k )^2}{2 \sigma^2}   }}
{\prod_{k=1}^{t}\frac{1}{\sqrt {2 \pi} \sigma} e^{ - \frac{X_k ^2}{2 \sigma^2}} }.
\end{align}

In this paper our goal is to find a computationally tractable form for the optimal, likelihood ratio test and also to characterize its asymptotic performance, when the number of samples $X_k$ grows large. In particular, with respect to performance characterization, we wish to compute the error exponent for the probability of a miss, under a given bound $\alpha$ on the probability of false alarm:
\begin{equation}\label{eq-error-exponent}
\lim_{t\rightarrow +\infty} - \frac{1}{t} \log P_{\mathrm{miss},t}^\alpha = :\zeta,
\end{equation}
where $P_{\mathrm{miss},t}^\alpha$ is the minimal probability of a miss among all decision tests that have probability of false alarm bounded by $\alpha$. By results from detection theory, e.g., ~\cite{Sung06},\cite{Chen96}, the $\zeta$ in~\eqref{eq-error-exponent} is given by the asymptotic Kullback-Leibler rate in~\eqref{eq-asymptotic-KL-rate}, provided that this limit exists
\begin{equation}
\label{eq-asymptotic-KL-rate}
\zeta= \lim_{t\rightarrow+\infty} - \frac{1}{t}\log L_t(X^t).
\end{equation}
We prove the existence of the limit in~\eqref{eq-asymptotic-KL-rate} in Lemma~\ref{lemma-FK-limit} in Section~\ref{sec-main} further ahead. An illustration of the identity~\eqref{eq-error-exponent} is given in Figure~\ref{Fig-asymptotic-KL-rate}, which clearly shows that both sequences $- \frac{1}{t} \log P_{\mathrm{miss},t}^\alpha$ and $- \frac{1}{t}\log L_t(X^t)$ are convergent and moreover that they converge to the same value -- the asymptotic Kullback-Leibler rate for the two hypothesis defined in~\eqref{eq-hypothesis-testing-problem}. For further details on this simulation see Section~\ref{sec-NumResults}.

\begin{figure}[htp]
	\centering
	\includegraphics[width=1\linewidth]{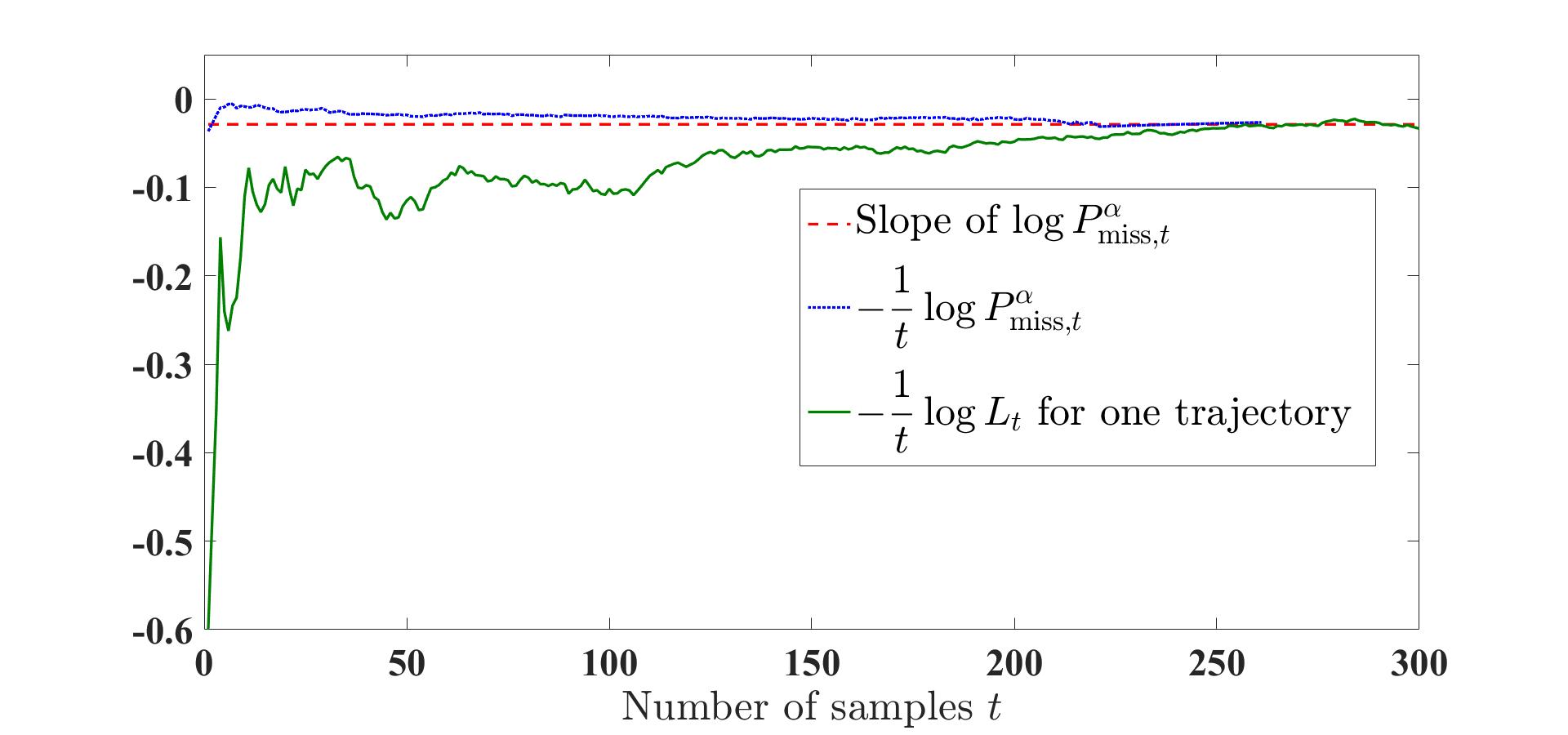}
	\caption{Simulation setup: $\Delta=3$, $p_1,p_2\sim\mathcal {U}([1,\Delta])$, $\mu_1=2$, $\mu_2=5$, $\sigma=10$, $\alpha=0.01$. Green full line plots the evolution of $-\frac{1}{t} \log L_t$; blue dotted line plots the evolution of $- \frac{1}{t} \log P_{\mathrm{miss},t}^\alpha$, and red dashed line plots the estimated slope of the probability of a miss values (in the logarithmic scale) calculated for values until $t=300$ observations.}
	\label{Fig-asymptotic-KL-rate}
\end{figure}

\section{Preliminaries}
\label{sec-Preliminaries}
In this section we now introduce a number of quantities related with the sequences $s^t\in \mathcal S^t$, $t=1,2,...$, and give certain results pertaining to these quantities that will be useful for our analysis.

\mypar{Statistics for the durations of phases} For each $t$, we define the sets of discrete times until time $t$ in which the signal was in states $1$ and $2$, which we respectively denote by $\mathcal T_1$ and $\mathcal T_2$:
\begin{align}\label{eq-def-tau}
\mathcal T_1 (s^t)&= \left\{ 1\leq k \leq t: \,s_k=1  \right\},\\
\mathcal T_2 (s^t)&= \left\{ 1\leq k \leq t: \,s_k=2  \right\}.
\end{align}
We denote cardinalities of $\mathcal T_1$ and $\mathcal T_2$, respectively, by $\tau_1$ and $\tau_2$, i.e., $\tau_1\equiv \left| \mathcal T_1\right|$ and $\tau_2\equiv \left| \mathcal T_2\right|$. Note that functions $\mathcal T_1$ and $\mathcal T_2$ are, strictly speaking, dependent on time $t$ (this dependence is observed in their domain sets $\mathcal S^t$ which clearly change with time $t$). However, for reasons of easier readibility, we suppress this dependence in the notation, as we also do for all the subsequently defined quantities.

For each $t$, for each $s^t$, we also introduce $N_{1}$ and $N_{2}$ to count the number of state-$1$ and state-$2$ phases, respectively, in the sequence $s^t$:
\begin{align}\label{eq-def-N-t}
N_{1} (s^t)& = \left|\left\{ 1\leq k \leq t: \,s_{k-1}=2,\,s_k=1  \right\}\right| \\
N_{2} (s^t) & = \left|\left\{ 1\leq k \leq t: \,s_{k-1}=1,\,s_k=2  \right\}\right|,
\end{align}
where, since the first phase is state-$1$ phase, we set $s_0\equiv 2$. We remark that, for any sequence $s^t$, if the last state $s_t=2$, then $N_1(s^t)=N_2(s^t)$, and if $s_t=1$, then $N_1(s^t)=N_2(s^t)+1$. Finally, $N (s^t)$ is the total number of phases in $s^t$, $N\equiv N_{1}+N_{2}$.

We further define the sets $\mathcal T_{mn} (s^t)$ that contain time indices for the $n$-th state-$m$ phase, $n=1,...,N_m(s^t)$, $m=1,2$. Note that, for each $m=1,2$, $\cup_{n=1}^{N_m(s^t)}\mathcal T_{mn} (s^t)= \mathcal T_m$. We now increase granularity in the counts $N_1$ and $N_2$ and define
\begin{align}\label{eq-counts-N1-N2-vectors}
N_{1d} (s^t)& = \sum_{n=1}^{N_1(s^t)} 1_{\left\{  \left| \mathcal T_{1n} \right| = d\right\}}(s^t),\mathrm{\;for\;}d=1,...,\Delta_1, \\
N_{2d} (s^t)& = \sum_{n=1}^{N_2(s^t)} 1_{\left\{  \left| \mathcal T_{2n} \right| =
d\right\}}(s^t),\mathrm{\;for\;}d=1,...,\Delta_2;
\end{align}
i.e., in words, vectors $\left(N_{m1},...,N_{m\Delta}\right)$, $m=1,2$, represent histograms of phase $1$ and phase $2$ durations. It is easy to see that $N_m=\sum_{d=1}^{\Delta_m} N_{md}$, for $m=1,2$. Also, for each time $t$ and each sequence $s^t$, the total number of state $1$ and state $2$ occurrences must sum up to $t$, and therefore $\sum_{d=1}^{\Delta_1} d\, N_{1d}(s^t)+ \sum_{d=1}^{\Delta_2} d\, N_{2d}(s^t)=t$.

Figure~\ref{fig:T4} shows an example of simulation signals under Hypothesis $\mathcal{H}_1$ with $\Delta=10,\mu_1=3,\mu_2=5$ and $\sigma=0.05$ using random duration model for various switching times $T$, difference process durations $D_{k,i}$ and numbers of different state-phases with fixed duration $N_{k,d}$. We can see from the figure that $D_{1,1}=T_1-T_0=8$ as shown in eq.~\eqref{def-state-1-durations} and there is only one state-phase $1$ last for $8$ samples, hence $N_{1,8}=1$. Again, from eq.~\eqref{def-state-2-durations} we can see from the figure again that $D_{2,1}=T_2-T_1=8$ and $D_{2,3}=T_6-T_5=8$. Thus $N_{2,8}=2$ for there are two state-phase $2$ last for $8$ samples.

\begin{figure}[htp]
	\centering
	\includegraphics[width=1\linewidth]{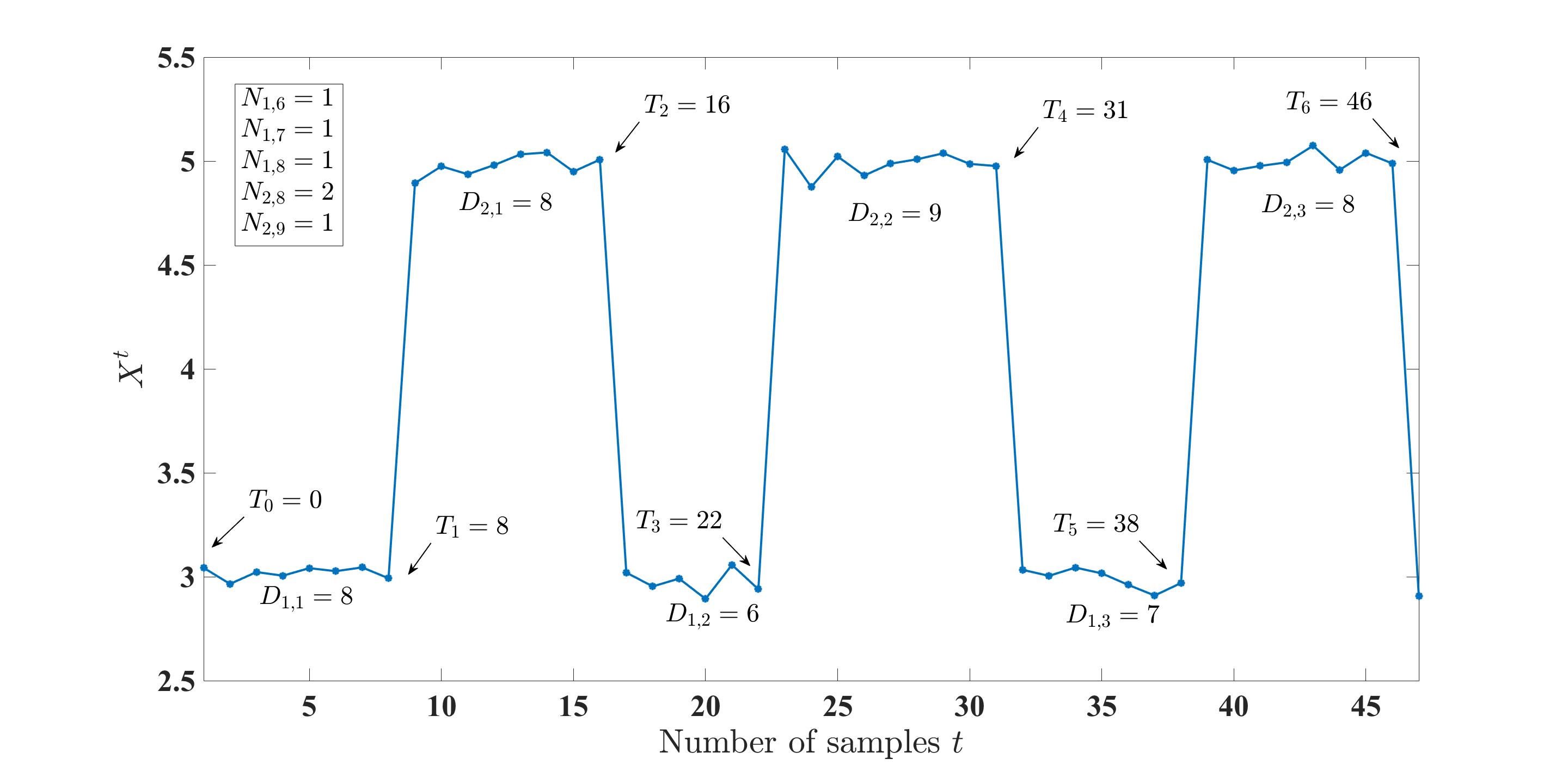}
	\caption{Example of simulation signals with $\Delta=10,\mu_1=3,\mu_2=5$ and $\sigma=0.05$ and various $T$, $D_{k,i}$, and $N_{k,d}$.}
	\label{fig:T4}
\end{figure}

To simplify the notation, let ${o}(s^t)$ return the duration of the last phase in the sequence $s^t$, and note also that $s_t$ returns the type of the last phase in $s^t$. The next lemma computes the probability of a given sequence $s^t$, $P(s^t)=\mathbb P_1\left(S^t=s^t\right)$.
\begin{lemma}
\label{lemma-sequence-s-t-probability}
For any sequence $s^t$, there holds
\begin{equation}
\label{eq-sequence-s-t-probability}
P(s^t)= \frac{p_{s_t\,{o}(s^t)}^{+} }{p_{s_t\,{o}(s^t)}}\, \prod_{d=1}^{\Delta_1}\, p_{1d}^{N_{1d} (s^t)} \, \prod_{d=1}^{\Delta_2} \,p_{2d}^{N_{2d} (s^t)},
\end{equation}
where by $p_{ml}^{+}$ we shortly denote $p_{ml}^{+} = p_{ml}+ p_{m l+1}+...+p_{m\Delta_m}$, for $l=1,2,...,\Delta_m$ and $m=1,2$.
\end{lemma}
The proof of Lemma~\ref{lemma-sequence-s-t-probability} is given in Appendix. Besides function $P$ which returns the exact probability of occurrence of sequence $s^t$, it will also be of interest to define a related function $P^\prime: \{1,2\}^t \mapsto \mathbb R$, defined through $P$ by leaving out the first factor in~\eqref{eq-sequence-s-t-probability}, i.e., $P^\prime (s^t)= \frac{{p_{s_t\,{o}(s^t)}}}{p_{s_t\,{o}(s^t)}^{+} } P(s^t)$ (note that the assumption  that $p_1,p_2>0$ (entrywise) ensures that $P^\prime$ is always well defined). Let $p_{\min}= \min\{p_{md}:\,m=1,2, \, d=1,...,\Delta_m\}$ and note that, for any $m$ and $d$, $p_{md}\leq  p_{md}^{+}\leq 1$ (this relation can be easily seen from the definition of $p_{md}^{+}$). Thus, the following relation holds between $P$ and $P^\prime$:
\begin{equation}\label{eq-relation-P-and-P-prime}
P^\prime(s^t) \leq  P(s^t)\leq \frac{1}{p_{\min}}\,P^\prime(s^t).
\end{equation}
For increasing $t$, the two functions will have equal exponents, that is, the effect of the factor $\frac{1}{p_{\min}}$ will vanish, and thus in our subsequent analyses we will use the analytically more appealing function $P^{\prime}$. Further, to simplify the analysis, in what follows we will assume that $\Delta_1=\Delta_2=: \Delta$.

We let $\mathcal S^t$ denote the set of all feasible sequences of states $s^t$ of length $t$, i.e., the sequences for which $\mathbb P_1(\mathcal S^t=s^t)>0$; we let $C_t$ denote the cardinality of $\mathcal S^t$. When $p_1$ and $p_2$ are strictly greater than zero, it can be shown that $C_t$ equals the number of ways in which integer $t$ can be partitioned with parts bounded by $\Delta$. This number is known as the $\Delta$-generalized Fibonacci number, and is computed via the following recursion:
\begin{equation}\label{eq-recursion-C-t}
 C_t=C_{t-1}+\ldots + C_{t-\Delta},
\end{equation}
with the initial condition $C_1=1$. The recursion in~\eqref{eq-recursion-C-t} is linear and hence can be represented in the form $\widetilde C_t = A \widetilde C_{t-1}$, where $\widetilde C_t= \left[C_t\, C_{t-1}\, \ldots \,C_{t-\Delta+1}\right]$ and $A$ is a square, $\Delta\times \Delta$ matrix; it can be shown that $A$ is equal to $A=e_1 {\mathbb 1}^\top + A_0$, where, we recall, $A_0$ is the lower shift matrix of dimension $\Delta$. The growth rate of $C_t$ is given by the largest zero of the characteristic polynomial of $A$, as the next result, which we borrow from~\cite{Flores67} asserts.
\begin{lemma}\label{lemma-C-t-growth-rate} [Asymptotics for $\Delta$-generalized Fibonacci number~\cite{Flores67}] For any $\epsilon$, there exists $t_0=t_0(\epsilon)$ such that
\begin{equation}\label{eq-C-t-growth-rate}
e^{t (\psi- \epsilon)} \leq  C_t\leq e^{t (\psi+ \epsilon)},
\end{equation}
where $\psi$ is the unique positive zero of the following polynomial $\psi^{\Delta} - \psi^{\Delta-1}-\ldots -1=0$.
\end{lemma}

\subsection{Sequence types}
\label{subsec-Sequence-types}

\mypar{Duration fractions} For $d=1,2,...,\Delta$, let $V_{m,d}$ denote the number of times along a given sequence of states that state-$1$ phase had length $d$, normalized by time $t$, i.e.,
\begin{equation}\label{eq-def-V-m-d}
V_{m,d}(s^t) = \frac{N_{m,d} (s^t)}{t}, \;m=1,2.
\end{equation}
For each sequence $s^t$, we define its type as the $2\times\Delta$ matrix $V:=\left( \left(V_1(s^t)\right)^\top; \left(V_2(s^t)\right)^\top \right)$, where $V_m(s^t)=\left( V_{m,1}(s^t),...,V_{m,\Delta}(s^t)\right)$, for $m=1,2$. Recalling $N_1$ and $N_2$~\eqref{eq-def-N-t}, which, respectively, count the number of state-$1$ and state-$2$ phases along $s^t$, we see that $N_m=t {\mathbb 1}^\top V_m$, $m=1,2$.

It will also be of interest to define the fractions of times $\Theta_1$ and $\Theta_2$ that a given sequence of states was in states $1$ and $2$, respectively,
\begin{equation}\label{eq-def-Theta-m}
\Theta_m(s_t) = \frac{\tau_m (s^t)}{t}, \;m=1,2.
\end{equation}
It is easy to verify that $\Theta_m = \sum_{d=1}^\Delta \,d\,V_{m,d}$, for $m=1,2$.

Let $\mathcal V_{t}$ denote the set of all $2\times \Delta$-tuples of feasible occurrence of type $V$ at time $t$
\begin{align}
\mathcal V_{t}& =\left\{\nu=(\nu_1,\nu_2): \nu=V(s^t), \mathrm{\;for\; some\;} s^t\right\}.
\end{align}
Note that, as they are defined as normalized versions of quantities $N_{md}(s^t)$, $V_{md}(s^t)$'s also inherit the properties of $N_{md}$'s: 1) $\sum_{d=1}^\Delta d V_{1d}(s^t)+ d V_{2d}(s^t)=1$; 2) $0 \leq {\mathbb 1}^\top V_1(s^t) - {\mathbb 1}^\top V_2(s^t) \leq 1/t$.  As $t\rightarrow +\infty$, for every $s^t \in \mathcal S^t$, the difference between ${\mathbb 1}^\top V_1(s^t)$ and ${\mathbb 1}^\top V_2(s^t)$ decreases. Motivated by this, we introduce the set
\begin{equation}\label{eq-def-mathcal-V}
\mathcal V= \left\{\nu\in \mathbb R_{+}^{2\times \Delta}: \, {\mathbb 1}^\top \nu_1={\mathbb 1}^\top \nu_2,\, q^\top \nu_1+ q^\top \nu_2 = 1 \right\}.
\end{equation}

For each $t$, $\nu \in \mathcal V_t$, define the set $\mathcal S^t_{\nu}$ that collects all sequences $s^t\in \mathcal S^t$ whose type is $\nu$:
\begin{equation}
\mathcal S^t_{\nu} = \left\{s^t\in \mathcal S^t:  V(s^t)=\nu\right\}
\end{equation}
(note that if $\nu \notin \mathcal V_t$, then set $\mathcal S^t_{\nu}$ would be empty). Set $\mathcal S^t_{\nu}$ therefore consists of all sequences with the following properties: 1) the first phase is state-$1$ phase; 2) the total number of state-$1$ phases is ${\mathbb 1}^\top \nu_1\,t$, where the total number of such phases of duration exactly $d$ is given by $\nu_{1,d}\,t$; and 3) the total number of state-$2$ phases is ${\mathbb 1}^\top \nu_2\,t$, where the total number of such phases of duration exactly $d$ is given by $\nu_{2,d}\,t$.

Let $C_{t,\nu}$ denote the cardinality of $\mathcal S^t_{\nu}$. This number is equal to the number of ways in which one can order ${\mathbb 1}^\top \nu_{1} t$ state-$1$ phases (of different durations), where each new ordering has to give rise to a different pattern of state occurrences, times the corresponding number for state-$2$ phases. Since for any $d$, any permutation of $\nu_{m,d} t$ phases, each of which is of length $d$, gives the same sequence pattern, $C_{t,\nu}$ is given by the number of permutations with repetitions for state-$1$ phases times the number of permutations with repetitions for state-$2$ phases:
\begin{equation}\label{eq-cardinality-C-t-nu}
C_{t,\nu} = \frac{\left( {\mathbb 1}^\top \nu_{1} t\right) ! }{ \left(\nu_{1,1} t\right)!\cdot \ldots \cdot \left(\nu_{1,\Delta_1} t\right) !} \frac{\left({\mathbb 1}^\top \nu_{2} t\right) ! }{ \left(\nu_{2,1} t\right)!\cdot \ldots \cdot \left(\nu_{2,\Delta_2} t\right) !}.
\end{equation}
From~\eqref{eq-cardinality-C-t-nu} the following result regarding the growth rate of $C_{t,\nu}$ easily follows (e.g., by Stirling's approximation bounds).

\begin{lemma}\label{lemma-growth-rate-C-t-nu}
For any $\epsilon>0$ there exists $t_1=t_1(\epsilon)$ such that for all  $t\geq t_1$
\begin{equation}\label{eq-growth-rate-C-t-nu}
e^{t \left( H\left(\nu_1\right)+ H\left(\nu_2\right)-\epsilon\right)} \leq C_{t,\nu} \leq e^{t \left( H\left(\nu_1\right)+ H\left(\nu_2\right)+\epsilon\right)},
\end{equation}
where $H:\mathbb R_+^\Delta \mapsto \mathbb R$ is defined as
\begin{equation}\label{eq-def-entropy}
H(\lambda) = - \sum_{d=1}^{\Delta} \frac{\lambda_d}{{\mathbb 1}^\top \lambda} \log \frac{\lambda_d}{{\mathbb 1}^\top \lambda},
\end{equation}
where $\lambda_d$ denotes the $d$-th element of an arbitrary vector $\lambda \in \mathbb R_+^\Delta$.
\end{lemma}

We end this section by giving some well-known results from the theory of large deviations that we will use in our analysis of detection problem~\eqref{eq-hypothesis-testing-problem}.

\subsection{Varadhan's lemma and large deviations principle}

\mypar{Large deviations principle}
\begin{definition}[Large deviations principle~\cite{DemboZeitouni93} with probability 1]
 \label{def-LDP-wp1}
  Let $\mu_t^\omega: \mathcal B\left( \mathbb R^D\right)$ be a sequence of Borel random measures defined on probability space $\left( \Omega,\mathcal F,\mathbb P\right)$. Then, $\mu_t^{\omega}$, $t=1,2,...$ satisfies the large deviations principle with probability one, with rate function $I$ if the following two conditions hold:
  \begin{enumerate}
    \item for every closed set $F$ there exists a set $\Omega_F^\star\subseteq \Omega$ with $\mathbb P\left(\Omega_F^\star\right)=1$, such that for each $\omega \in \Omega_F^\star$,
    \begin{equation}\label{def-LDP-upper}
      \limsup_{t\rightarrow +\infty} \frac{1}{t}\log  \mu^\omega_t(F) \leq - \inf_{x\in F} I(x);
    \end{equation}
    \item for every open set $E$ there exists a set $\Omega_E^\star\subseteq \Omega$ with $\mathbb P\left(\Omega_E^\star\right)=1$, such that for each $\omega \in \Omega_E^\star$,
    \begin{equation}\label{def-LDP-lower}
      \liminf_{t\rightarrow +\infty} \frac{1}{t} \log \mu^\omega_t(E) \geq - \inf_{x\in E} I(x).
    \end{equation}
  \end{enumerate}
\end{definition}

We give here the version of the Varadhan's lemma which involves sequence of random probability measures and large deviations principle (LDP) with probability one.

\begin{lemma} [Varadhan's lemma~\cite{DemboZeitouni93}]
\label{lemma-Varadhans-wp1}
Suppose that the random sequence of measures $\mu^\omega_t$ satisfies the LDP with probability one, with rate function $I$, see Definition~\ref{def-LDP-wp1}. Then, if for function $F$ the tail condition below holds with probability one,
\begin{equation}\label{eq-tail-condition}
\lim_{B \rightarrow + \infty} \limsup_{t\rightarrow +\infty} \frac{1}{t}\,\log \int_{x: F(x)\geq B } e^{t F(x)} d\mu_t^\omega(x) = -\infty,
 \end{equation}
 then, with probability one,
 \begin{equation}\label{eq-Varadhan}
\lim_{t\rightarrow +\infty} \frac{1}{t} \log \int_{x} e^{t F(x)} d\mu_t^\omega(x) = \sup_{x\in \mathbb{R}^D} F(x) - I(x).
 \end{equation}
\end{lemma}

\section{Linear recursion for the LLR and the Lyapunov exponent}
\label{sec-LLR-recursion}
From~\eqref{eq-LLR} and~\eqref{eq-sequence-s-t-probability}, it is easy to see that the likelihood ratio can be expressed through the defined quantities as:
\begin{align}
&\!\! L_t(X^t) = \sum_{s^t \in \mathcal S^t}
P(s^t) e^{ \frac{1}{\sigma^2} \sum_{m=1}^2\mu_m  \sum_{k \in \mathcal T_m (s^t)}  X_{k}  - \tau_m (s^t) \frac{\mu_m^2}{2 \sigma^2}  }\nonumber\\
&= \sum_{s^t \in \mathcal S^t} \frac{p_{s_t,{o}(s^t)}^{+} }{p_{s_t,{o}(s^t)}}
 e^{\sum_{m=1}^2\sum_{d=1}^{\Delta_m} N_{1m}(s^t) \log p_{1m}}\times \nonumber\\
&\;\;\;\; e^{ \frac{1}{\sigma^2} \sum_{m=1}^2 \mu_m  \sum_{k \in \mathcal T_m (s^t)}  X_{k}  - \tau_m (s^t) \frac{\mu_m^2}{2 \sigma^2}  }\label{eq-LLR-2}.
\end{align}

 The expression in~\eqref{eq-LLR-2} is combinatorial, and its straightforward implementation would require computing $C_t\approx e^{\psi t}$ summands. This is prohibitive when the observation interval $t$ is large. In this paper, we unveil a simple, linear recursion form for the likelihood $L_t(X^t)$, for $t=1,2,...$. We give this result in the next lemma. To shorten the notation, we introduce functions $f_m: \mathbb R \mapsto \mathbb R$, which we define by $f_m(x):= \frac{1}{\sigma^2}\mu_m x - \frac{1}{2 \sigma^2} \mu_m^2$, for $x\in \mathbb R$ and $m=1,2$. Recall that $e_1$ denotes the first canonical vector in $\mathbb R^\Delta$ (the $\Delta$ dimensional vector with $1$ only in the first position, and having zeros in all other positions), and $1$ denotes the vector of all ones in $\mathbb R^\Delta$.

\begin{lemma}\label{lemma-recursion-for-LRT} Let $\Lambda_k = \left( {\Lambda^1_{k}}^\top,  {\Lambda^2_{k}}^\top\right)^\top$ evolve according to the following recursion
\begin{equation}\label{eq-Lambda-recursion}
\Lambda_{k+1} = A_{k+1} \Lambda_k,
\end{equation}
with the initial condition
$\Lambda_1 = \left( e^{f_1 (X_{k})}e_1^\top,\, e^{f_2 (X_{k})}e_1^\top\right)^\top$, and where, for $k\geq 2$, matrix $A_k= [A^{11}_k A^{12}_k; A^{21}_k A^{22}_k]$ is defined by
\begin{align}
\label{eq-matrices-A-k}
A^{11}_k & = e^{f_1 (X_{k})} A_0 \nonumber\\
A^{12}_k & = e^{f_1 (X_{k})}  e_1 p_2^\top \nonumber\\
A^{21}_k & = e^{f_2 (X_{k})}  e_1 p_1^\top \nonumber\\
A^{22}_k & = e^{f_2 (X_{k})} A_0,
\end{align}
and $A_0$ is, we recall, the lower shift matrix of dimension $\Delta$. Then, the likelihood ratio $L_t(X^t)$ is, for each $t\geq 1$, computed by
\begin{equation}\label{eq-LLR-lemma}
L_t(X^t) = \sum_{d=1}^\Delta   p_{1d}^{+} {\Lambda}^1_{t,d} +  p_{2d}^{+} {\Lambda}^2_{t,d},
\end{equation}
where $\Lambda^m_{t,d}$ is the $d$-th element of $\Lambda^m_{t}$, for $d=1,...,\Delta$ and $m=1,2$.
\end{lemma}
\mypar{Remark} We note that the matrix $A_k$ can be further decomposed as
\begin{align} \label{eq-D-k-and-M-decomposition}
A_k& =D_k M_0\\
D_k & = \mathrm{diag} \left(\left( e^{f_1(X_k)} {\mathbb 1}^\top,\,e^{f_2(X_k)} {\mathbb 1}^\top\right)^\top\right),\,\;k=1,2,...,\\
M_0 & = \left[
\begin{array}{ll}
A_0& e_1 p_2^\top\\
e_1 p_1^\top & A_0
\end{array}
\right],
\end{align}
i.e., $D_k$ is a random diagonal matrix of size $2\Delta$, modulated by the $k$-th measurement $X_k$, and $M_0$ is a sparse, constant matrix of the same dimension, which defines transitions from the current state pattern to the one in the next time step.

\mypar{Proof intuition} The intuition behind this recursive form is the following. We break the sum in~\eqref{eq-LLR-2} into sequences $s^t$ whose last phases are of the same type.  For sequences that end with state $m=1$, ${\Lambda}^1_{t,d}$ represents the contribution to the overall likelihood ratio $L_t(X^t)$ of all such sequences whose last phase is of length $d$, and similarly for ${\Lambda}^2_{t,d}$. Once the vectors ${\Lambda}^1_{t,d}$ and ${\Lambda}^2_{t,d}$ are defined, their update is simple. Consider the value ${\Lambda}^1_{t+1,d}$, where $d>1$; this value corresponds to the likelihood ratio contribution of all sequences $s^{t+1}$ that end with state-$1$ phase of duration $d$. Since $d>1$, the only possible way to get a sequence of that form is to have a sequence at time $t$ that ends with the same state, where the duration of the last phase is $d-1$. This translates to the update ${\Lambda}^1_{t+1,d}= e^{f_1 (X_{t+1})}{\Lambda}^1_{t,d-1}$, where the choice of $f_1$ in the exponent is due to the fact that the last state is $s_{t+1}=1$; see also the first line in~\eqref{eq-matrices-A-k}. On the other hand, if $d=1$, then the state at time $t$ must have been $m=2$. The duration of this previous phase could have been arbitrary from $d=1$ to $d=\Delta$. Hence ${\Lambda}^1_{t+1,1}$ is computed as the sum $\Lambda^1_{t+1,1}= \sum_{d=1}^\Delta p_{2d} e^{f_1 (X_{t+1})} {\Lambda}^2_{t,d}$, where the probabilities $p_{2d}$ are used to mark that the previous phase is completed, see the second line in~\eqref{eq-matrices-A-k}. The analysis for ${\Lambda}^2_{t+1,d}$ is similar. The formal proof of Lemma~\ref{lemma-recursion-for-LRT} is given in Appendix.


\subsection{Error exponent~$\zeta$ as Lyapunov exponent} From Lemma~\ref{lemma-recursion-for-LRT} we see that $L_t$ can be represented as a linear function of the matrix product $\Pi_t:=A_t\cdot\ldots\cdot A_1$,
\begin{equation}
\label{eq-LRT-as-FK-product}
L_t= {p^{+}}^\top \Pi_t \Lambda_0,
\end{equation}
where $A_k$ are matrices of the form~\eqref{eq-matrices-A-k}, and $p^{+}=\left[ {p_1^{+}}^\top,\,{p_2^{+}}^\top\right]^\top$, where the $d$-th entry of $p_m^{+}$ equals $p_{md}^{+}$, for $m=1,2$, $d=1,2,...,\Delta$. Each $A_k$ is modulated by the measurement $X_k$ obtained at time $k$. Since $X_k$'s, $k=1,2,...$, are i.i.d., it follows that the matrices $A_k$ are i.i.d. as well. Applying a well-known result from the theory of random matrices, see Theorem~2 in~\cite{Furstenberg1960}, to sequence $A_k$ it follows that the sequence of the negative values of the normalized log-likelihood ratios $-\frac{1}{t}\log L_t$, $t=1,2,...$, converges to the Lyapunov exponent of the matrix product $\Pi_t$. This result is given in Lemma~\ref{lemma-FK-limit} and proven in Appendix.

\begin{lemma}\label{lemma-FK-limit}
With probability one,
\begin{equation}\label{eq-Lyapunov-limit-expectations}
\lim_{t\rightarrow+\infty} \frac{1}{t}\log \|\Pi_t\| = \lim_{t\rightarrow+\infty} \frac{1}{t} \mathbb E_0\left[\log \|\Pi_t\|\right],
\end{equation}
and thus, with probability one,
\begin{equation}\label{eq-Lyapunov-limit}
\zeta = \lim_{t\rightarrow+\infty} - \frac{1}{t}\log \|\Pi_t\| = \lim_{t\rightarrow+\infty} - \frac{1}{t} \mathbb E_0\left[\log L_t\right].
\end{equation}
\end{lemma}

Lemma~\ref{lemma-FK-limit} asserts that the error exponent for hypothesis testing problem~\eqref{eq-hypothesis-testing-problem} equals the top Lyapunov exponent for the sequence of products $\Pi_t$. Computation of the Lyapunov exponent (e.g., for i.i.d. matrices) is a well-known problem in random matrix theory and theory of random dynamical systems, proven to be very difficult to solve, see, e.g.,~\cite{TsitsiBlondel97}. We instead search for tractable lower bounds that tightly approximate $\zeta$. We base our method for approximating $\zeta$ on the right hand-side identity in~\eqref{eq-Lyapunov-limit}.

\section{Main result}
\label{sec-main}

Our first step for computing the limit in~\eqref{eq-Lyapunov-limit} is a natural one. Since $\mu_1\geq 0$ is the guaranteed signal level (recall that $\mu_2>\mu_1\geq 0$), we assume that the signal was at all times at state $1$, and remove the corresponding components of the signal to noise ratio (SNR) $\frac{\mu_1^2}{2\sigma^2}$ and the signal sum $\sum_{k=1}^t X_k$ from the likelihood ratio. This manipulation then gives us a lower bound on the error exponent. By doing so, we arrive at an equivalent problem to problem~\eqref{eq-hypothesis-testing-problem} just with $\mu_1=0$. Mathematically, we have
\begin{align}
\label{eq-LLR-3}
L_t(X^t)\!\! &=\!\! \sum_{s^t \in \mathcal S^t} P(s^t) e^{ \frac{1}{\sigma^2} \mu_1
\left(\sum\limits_{k =1}^t  X_{k} - \sum\limits_{k \in \mathcal T_2(s^t)}  X_{k} - (t-\tau_2(s^t))\frac{\mu_1^2}{2 \sigma^2}\right)}  \times \nonumber\\
&\times e^{  \frac{1}{\sigma^2}  \mu_2 \sum\limits_{k \in \mathcal T_2 (s^t)}   X_{k}  - \tau_2 (s^t) \frac{\mu_2^2}{2\sigma^2} } = e^{ \frac{1}{\sigma^2} \mu_1  \sum\limits_{k =1}^t  X_{k}  - t \frac{\mu_1^2}{2 \sigma^2}}\times \nonumber\\
&\times \sum_{s^t \in \mathcal S^t} P(s^t) e^{ \frac{1}{\sigma^2}  \sum\limits_{k \in \mathcal T_2 (s^t)}  (\mu_2 -\mu_1) X_{k}  - \tau_2 (s^t) \frac{\mu_2^2 - \mu_1^2}{2\sigma^2} }.
\end{align}
Taking the logarithm, dividing by $t$, and computing the expectation with respect to hypothesis $\mathcal H_0$, we get
\begin{align}
\label{eq-LLR-4}
& \!\!\!\frac{1}{t}\mathbb E_0\left[\log L_t(X^t)\right]= - \frac{\mu_1^2}{2 \sigma^2} + \frac{1}{t}\mathbb E_0\left[\log \sum_{s^t \in \mathcal S^t} P(s^t)\times\right.\nonumber \\
&\;\;\;\;\;\times\left. e^{ \frac{1}{\sigma^2}  \sum_{k \in \mathcal T_2 (s^t)}  (\mu_2 -\mu_1) X_{k}  - \tau_2 (s^t) \frac{\mu_2^2 - \mu_1^2}{2\sigma^2} }\right],
\end{align}
where we used that $\mathbb E_0\left[X_k\right]=0$, for all $k$, see~\eqref{eq-hypothesis-testing-problem}. Taking the limit as $t\rightarrow +\infty$, we obtain
\begin{equation}\label{eq-new-limit}
\zeta = \frac{\mu_1^2}{2 \sigma^2} + \eta,
\end{equation}
where $\eta$ is given by the following limit
\begin{align}
\label{eq-def-eta}
\eta= \lim_{t\rightarrow +\infty} &-\frac{1}{t}\mathbb E_0\left[ \log \sum_{s^t \in \mathcal S^t} P(s^t)\times\right.\nonumber \\
&\left.\times e^{ \frac{1}{\sigma^2}  \sum_{k \in \mathcal T_2 (s^t)}  (\mu_2 -\mu_1) X_{k}  - \tau_2 (s^t) \frac{\mu_2^2 - \mu_1^2}{2\sigma^2} }\right],
\end{align}
the existence of which is guaranteed by~\eqref{eq-Lyapunov-limit}, in Lemma~\ref{lemma-FK-limit}. From now on, we focus on computing $\eta$. Before we proceed, we make a simplification in the expression for $\eta$ by replacing the term $P(s^t)$ with its analytically more appealing proxy $P^\prime(s^t)$, see~\eqref{eq-relation-P-and-P-prime}. Applying inequality~\eqref{eq-relation-P-and-P-prime} in~\eqref{eq-def-eta} and using the fact that $\frac{1}{t}\log p_{\min}\rightarrow 0$, as $t\rightarrow +\infty$, we obtain that the limit in~\eqref{eq-def-eta} does not change when we replace $P(s^t)$ with $P^\prime(s^t)$, i.e.,
\begin{align}
\label{eq-def-eta-2}
\eta= \lim_{t\rightarrow +\infty} &-\frac{1}{t}\mathbb E_0\left[ \log \sum_{s^t \in \mathcal S^t} P^\prime(s^t)\times\right.\nonumber \\
&\left.\times e^{ \frac{1}{\sigma^2}  \sum_{k \in \mathcal T_2 (s^t)}  (\mu_2 -\mu_1) X_{k}  - \tau_2 (s^t) \frac{\mu_2^2 - \mu_1^2}{2\sigma^2} }\right].
\end{align}

For $\lambda \in \mathbb R^{\Delta}$, and $p\in \mathbb S^{\Delta-1}$, introduce the relative entropy function $D(\lambda||p):= \sum_{d=1}^\Delta \frac{\lambda_d}{{\mathbb 1}^\top \lambda}
\log {\frac{\lambda_d}{{\mathbb 1}^\top \lambda}p_d}$.
\begin{theorem}
\label{theorem-main}
There holds $\underline \eta + \frac{\mu_1^2}{2\sigma^2}\leq \zeta$, where $\underline \eta$ is the optimal value of the following optimization problem
\begin{equation}
\begin{array}[+]{lc}
\mathrm{minimize} &  G(\nu)\\
\mathrm{subject\; to} & H(\nu_1)+ H(\nu_2) \geq \frac{\xi^2}{2\theta_2\sigma^2}\\
& \theta_2 = q^\top \nu_2 \\
& \nu \in \mathcal V\\
& \xi \in \mathbb R.
\end{array},
\label{eq-main-LB-opt}
\end{equation}
where $G(\nu)=D(\nu_1|| p_1) +  D(\nu_2||p_2) +
\frac{\theta_2}{2 \sigma^2}  \left( \frac{\xi}{\theta_2}-(\mu_2-\mu_1)\right)^2  + \theta_2 \frac{  \mu_1(\mu_2-\mu_1)}{\sigma^2}$, for $\nu \in \mathbb R_{+}^{2\Delta}$, $\xi\in \mathbb R$.
\end{theorem}

\mypar{Guaranteed error exponent} Since each of the terms in the objective function of~\eqref{eq-main-LB-opt} is non-negative, its optimal value is lower bounded by $0$. Using relation~\eqref{eq-new-limit}, we obtain that the value of the error exponent is lower bounded by the value of SNR in state-$1$, $\frac{\mu_1^2}{2\sigma^2}$, i.e.,
\begin{equation}\label{eq-guaranteed-exponent}
\zeta \geq \frac{\mu_1^2}{2 \sigma^2}.
\end{equation} The preceding bound holds for any choice of parameters $\Delta, p_1,p_2,\mu_1$ and $\mu_2$. This result is very intuitive, as it mathematically formalizes the reasoning that, no matter which configuration of states occurs, signal level $\mu_1$ is always guaranteed, and hence the corresponding value of error exponent $\frac{\mu_1^2}{2 \sigma^2}$ is ensured. In that sense, any appearance of state $2$ (i.e., signal level $\mu_2>\mu_1$) can only increase the error exponent.

\mypar{Special case $\mu_1=0$ and detectability condition} When the signal level in state $1$ equals zero, then, since the statistics of $X_k$ for $S_k=1$ is the same as its statistics under $\mathcal H_0$, effectively we can have information on the state of nature $\mathcal H_1$ only when state $S_k=2$ occurs. Denoting $\mu=\mu_2$, optimization problem~\eqref{eq-main-LB-opt} then simplifies to:
\begin{equation}
\begin{array}[+]{lc}
\mathrm{minimize} &  D(\nu_1|| p_1) +  D(\nu_2||p_2) +
\frac{\theta_2}{2  \sigma^2}  \left( \frac{\xi}{\theta_2}-\mu\right)^2  \\
\mathrm{subject\; to} & H(\nu_1)+ H(\nu_2) \geq \frac{\xi^2}{2\theta_2\sigma^2}\\
& \theta = \nu^\top q \\
& \nu \in \mathcal V\\
& \xi \in \mathbb R.
\end{array}.
\label{eq-main-LB-opt-2}
\end{equation}

From~\eqref{eq-main-LB-opt-2} we obtain the following condition for detectability of process $S_k$:
\begin{equation}\label{eq-detectability-condition}
H(p_1) + H(p_2) \geq \frac {q^\top p_2} {q^\top p_1+q^\top p_2}\frac{ \mu^2}{2\sigma^2},
\end{equation}
i.e., if the inequality above holds, then the optimal value of optimization problem~\eqref{eq-main-LB-opt-2} is zero. To see why this holds, note that the point $(\nu_1,\nu_2,\xi)\in \mathbb R^{2\Delta+1}$, where $\nu_m=p_m/(q^\top p_1+q^\top p_2)$, $m=1,2$, and $\xi=q^\top p_2/((q^\top p_1+q^\top p_2)) \mu$ under which the cost function of~\eqref{eq-main-LB-opt-2} vanishes, under condition~\eqref{eq-detectability-condition} belongs to the constraint set of~\eqref{eq-main-LB-opt-2}. Thus, under condition~\eqref{eq-detectability-condition}, the lower bound on the error exponent $\underline \eta$ is zero, indicating that the process $S_k$ is not detectable. To further illustrate this condition, note that the left hand-side corresponds to the entropy of the process $S_k$, and the right hand-side corresponds to the expected, i.e. -- long-run SNR of the measured signal ($q^\top p_2 / \left(q^\top p_1+q^\top p_2\right)$ is the expected fraction of times that the process was in state $2$, and $\frac{\mu^2}{2\sigma^2}$ is the SNR for this state). Condition~\eqref{eq-detectability-condition} therefore asserts that, if the entropy of the process $S_k$ is too high compared to the expected, or long-run, SNR, then it is not possible to detect its presence. Intuitively, if the dynamics of the phase durations is too stochastic, then it is not possible to estimate the locations of state $2$ occurrences, in order to perform the likelihood ratio test. However, on the other hand, if the SNR is very high (e.g., the level $\mu$ is high compared to the process noise $\sigma^2$) then, whenever state $2$ occurs, the signal will make a sharp increase and can therefore be easily detected. The condition in this sense quantitatively characterizes the threshold between the two physical quantities which makes detection possible.

\subsection{Reformulation of~\eqref{eq-main-LB-opt-2}}

In this subsection we show that optimization problem~\eqref{eq-main-LB-opt-2} admits a simplified form, obtained by suppressing the dependence on $\xi$ through inner minimization over this variable. 
To simplify the notation, introduce $H(\nu)=H(\nu_1)+H(\nu_2)$ and $R(\nu) = q^\top \nu_2 \frac{\mu^2}{2\sigma^2}$; note that the function $R$ has the physical meaning of the expected SNR of the $S_t$ process that we wish to detect, for a given sequence type $\nu$.

\begin{lemma}\label{lemma-convex-reformulation}
Suppose that $H(p_1)+H(p_2)< q^\top p_2 / \left(q^\top p_1+q^\top p_2\right) \frac{\mu^2} {2 \sigma^2}$. Then, optimization problem~\eqref{eq-main-LB-opt-2} is equivalent to the following optimization problem:
\begin{equation}
\begin{array}[+]{lc}
\mathrm{minimize} &  D(\nu_1|| p_1) +  D(\nu_2||p_2) +   \left( \sqrt{ H(\nu)} - \sqrt{R(\nu)} \right)^2 \\
\mathrm{subject\; to} & H(\nu) \leq  R(\nu)\\
& \nu \in \mathcal V\\
\end{array}.
\label{eq-convex-reformulation}
\end{equation}
\end{lemma}

\begin{proof}
 Fix $\nu\in \mathcal V$. To remove the dependence on $\xi$ in~\eqref{eq-main-LB-opt-2}, for any given fixed $\nu\in \mathcal V$, we need to solve
\begin{equation}\label{eq-inner-minimization}
\begin{array}[+]{lc}
\mathrm{minimize} & \theta_2 \frac{ \left( \frac{\xi}{\theta_2} - \mu\right)^2}{2\sigma^2}\\
\mathrm{subject\; to} & H(\nu) \geq \frac{\xi^2}{2\theta_2\sigma^2}\\
& \xi\in \mathbb R
\end{array},
\end{equation}
where, as before, we denote $\theta_2=q^\top \nu_2$. Since $\mu>0$, and the constraint set is defined only through the square of $\xi$, the optimal solution of~\eqref{eq-inner-minimization} is achieved for $\xi\geq 0$.  Thus,~\eqref{eq-inner-minimization} is equivalent to
\begin{equation}\label{eq-inner-minimization-2}
\begin{array}[+]{lc}
\mathrm{minimize} & \theta_2 \frac{ \left( \frac{\xi}{\theta_2} - \mu\right)^2}{2\sigma^2}\\
\mathrm{subject\; to} & 0\leq \xi \leq \sigma \sqrt{ 2\theta_2H(\nu)}
\end{array}.
\end{equation}
The solution of~\eqref{eq-inner-minimization-2} is given by: 1) $\xi^\star = \theta_2\mu$, if $\theta_2 \mu \leq \sigma \sqrt{ 2\theta_2 H(\nu)}$; and 2) $\xi^\star=\sigma \sqrt{ 2\theta_2H(\nu)}$, otherwise. Hence, to solve~\eqref{eq-main-LB-opt-2} we can partition its constraint set $\mathcal V= \mathcal V_1\bigcup\mathcal V_2$ according to these two cases, where $\mathcal V_1=\left\{\nu \in\mathcal V:\, H(\nu) \geq \theta_2\frac{\mu^2}{2\sigma^2}\right\}$ and $\mathcal V_2=\left\{\nu \in\mathcal V:\, H(\nu) \leq \theta_2\frac{\mu^2}{2\sigma^2}\right\}$, solve the corresponding two optimization problems, and finally find the minimum among the two obtained optimal values.

Consider first the case $\nu \in \mathcal V_1$. Since in this case $\xi^\star = \theta_2\mu$, plugging in this value in~\eqref{eq-inner-minimization-2}, we have that the optimization problem~\eqref{eq-main-LB-opt-2} with $\mathcal V$ reduced to $\mathcal V_1$ simplifies to:
\begin{equation}
\begin{array}[+]{lc}
\mathrm{minimize} &  D(\nu_1|| p_1) +  D(\nu_2||p_2)\\
\mathrm{subject\; to} & \nu \in \mathcal V_1.
\end{array}.
\label{eq-opt-mathcal-V-1}
\end{equation}
If $H(p) \geq \frac{q^\top p_2}{q^\top p_1+q^\top p_2} \frac{\mu^2}{2\sigma^2}$, then the point $1/ \left(q^\top p_1+q^\top p_2\right)p $ belongs to $\mathcal V$, where $p=(p_1,p_2)$ and hence the optimal solution to~\eqref{eq-opt-mathcal-V-1} equals $1/ \left(q^\top p_1+q^\top p_2\right)p$ with the corresponding optimal value equal to $0$. Suppose now that $H(p) < \frac{q^\top p_2}{q^\top p_1+q^\top p_2} \frac{\mu^2}{2\sigma^2}$. We show that in this case the solution to~\eqref{eq-opt-mathcal-V-1} must be at the boundary of the constraint set, in the set of points $\left\{ \nu\in \mathcal V:\,H(\nu) = \theta_2\frac{\mu^2}{2\sigma^2}\right\}$.

We prove the above claim. Since the entropy function $H$, see eq.~\eqref{eq-def-entropy}, is concave, the constraint set $\mathcal V_1$ is convex, and since KL divergence $D$ is convex, we conclude that the problem in~\eqref{eq-opt-mathcal-V-1} is convex. Also, it can be shown that the Slater point exists~\cite{BoydsBook04}. Therefore, the solution to~\eqref{eq-opt-mathcal-V-1} is given by the corresponding Karush-Kuhn-Tucker (KKT) conditions:
\begin{equation}
\label{eq-KKT-cvx-1}
\left\{\begin{array}{l}
(1+\lambda) \log\frac{\nu_{1d}}{ {\mathbb 1}^\top \nu_1} - \log p_{1d}= 0,\,\,\mathrm{for}\;d=1,...,\Delta\\
(1+\lambda) \log\frac{\nu_{2d}}{{\mathbb 1}^\top \nu_2} - \log p_{2d} +\lambda d \frac{\mu^2}{2\sigma^2}= 0,\,\,\mathrm{for}\;d=1,...,\Delta\\
H(\nu) \geq  q^\top \nu_2 \frac{\mu^2}{2\sigma^2} \\ 
\lambda\geq 0\\
\lambda\, \left(H(\nu) - q^\top \nu_2 \frac{\mu^2}{2\sigma^2}\right)= 0\\
\nu \in \mathcal V
\end{array}\right..
\end{equation}
From the fourth and fifth condition, we have that either $\lambda =0$, or that $\lambda>0$ and $H(\nu)= q^\top \nu_2 \frac{\mu^2}{2\sigma^2}$. Suppose that $\lambda=0$. Then, from the first two KKT conditions we have that the solution $\nu$ must satisfy $\nu_{md}/{\mathbb 1}^\top \nu_m=p_{md}$, for $m=1,2$, $d=1,...,\Delta$. However, this contradicts with the third condition (recall that we assumed that $H(p) < q^\top p_2 \frac{\mu^2}{2}$). Therefore, the solution to~\eqref{eq-opt-mathcal-V-1} must belong to the set $\left\{ \nu\in \mathcal V:\, H(\nu)= q^\top p_2/\left(q^\top p_1+q^\top p_2 \right) \frac{\mu^2}{2\sigma^2}\right\}$. Since this set intersects with the set $\mathcal V_2$, we conclude that, when $H(p) < q^\top p_2/\left(q^\top p_1+q^\top p_2 \right) \frac{\mu^2}{2\sigma^2}$, then the optimal solution to~\eqref{eq-main-LB-opt-2} is found by optimizing over the smaller set $\mathcal V_2\subseteq \mathcal V$, i.e.,~\eqref{eq-main-LB-opt-2} is equivalent to

\begin{equation}
\begin{array}[+]{lc}
\mathrm{minimize} &  D(\nu_1|| p_1) +  D(\nu_2||p_2) + \frac{\theta_2}{2\sigma^2}\left( \frac{\xi^\star}{\theta_2}  - \mu \right)^2 \\
& \nu \in \mathcal V_2.
\end{array},
\label{eq-reformulation}
\end{equation}
where $\xi^\star (\nu)= \sigma \sqrt{ 2\theta_2 H(\nu)}$. Simple algebraic manipulations reveal that the third term in the objective above is equal to $\left(\sqrt{H(\nu)}  - \sqrt{R(\nu)} \right)^2$. Finally, set $\mathcal V_2$ is precisely the constraint set in~\eqref{eq-main-LB-opt-2}, and hence the claim of the lemma follows.
\end{proof}


\section{Proof of Theorem~\ref{theorem-main}}
\label{sec-proof-main}

\mypar{Sum of conditionals as an expectation} For each $s^t\in\mathcal S_t$, introduce
\begin{equation}\label{eq-def-mathcal-X}
\mathcal X_{s^t} =\frac{1}{t} \sum_{k\in \mathcal T_2} X_k,
\end{equation}
and note that, for each $s^t$ and under $\mathcal H=\mathcal H_0$, $\mathcal X_{s^t}$ is Gaussian random variable of mean zero and variance equal to $\tau_2(s^t)/t^2= \theta_2(s^t)/t$. The idea is to view the sum in~\eqref{eq-def-eta-2} as an expectation of a certain function $g_\mathcal X: \mathcal S_t\mapsto \mathbb R$ defined over the set $\mathcal S_t$ of all possible sequences $s^t$, parameterized by random family (i.e., vector) $\mathcal X=\left\{\mathcal X_{s^t}:\,s^t\in \mathcal X^t\right\}$. More precisely, consider the probability space with the set of outcomes $\mathcal S_t$ and where an element $s^t$ of $\mathcal S_t$ is drawn uniformly at random -- and hence with probability $1/C_t$, where, we recall $C_t= |\mathcal S^t|$; denote the corresponding expectation by $\mathbb E_{U}$. We see that the sum under the logarithm in~\eqref{eq-def-eta-2} equals
\begin{align}\label{eq-expectation-1}
&\!\!\sum_{s^t \in \mathcal S^t} P^\prime(s^t) e^{ t\frac{(\mu_2 -\mu_1)}{\sigma^2}   \mathcal X_{s^t}  - \tau_2 (s^t) \frac{\mu_2^2 - \mu_1^2}{2\sigma^2} } \nonumber\\
&=C_t \sum_{s^t \in \mathcal S^t}  \frac{1}{C_t} g_{\mathcal X}(s^t) = C_t\, \mathbb E_U\left[ g_{\mathcal X}(s^t) \right],
\end{align}
where it is easy to see that $g_{\mathcal X}(s^t)=P^\prime(s^t) e^{ t\frac{ (\mu_2 -\mu_1)}{\sigma^2} \mathcal X_{s^t}  - \tau_2 (s^t) \frac{\mu_2^2 - \mu_1^2}{2\sigma^2} }$, for $s^t\in \mathcal S_t$.

Using further the type $V$ defined in Subsection~\ref{subsec-Sequence-types}, we can express $g_{\mathcal X}(s^t)$ as
\begin{equation}
\label{eq-def-g}
g_{\mathcal X}(s^t) =  e^{  t\frac{ (\mu_2 -\mu_1)}{\sigma^2} \mathcal X_{s^t}  - t \Theta_2 (s^t) \frac{\mu_2^2 - \mu_1^2}{2\sigma^2} + t \sum\limits_{m=1}^2 \sum\limits_{d=1}^\Delta  V_{md} (s^t) \log p_{md}}.
\end{equation}

\mypar{Induced measure} We see that function $g_{\mathcal X}$ depends on $s^t$ only through type $V$ of the sequence and the values of vector $\mathcal X$. More precisely, define $F: \mathbb R^{2\Delta}\times \mathbb R\mapsto \mathbb R$ as
\begin{equation}\label{eq-def-F}
F(\nu,\xi)= \frac{\mu_2 -\mu_1}{\sigma^2}\xi -\theta_2 \frac{\mu_2^2 - \mu_1^2}{2\sigma^2}
+ \sum_{m=1}^2\sum_{d=1}^\Delta  \nu_{md} \log p_{md}.
 \end{equation}
Then, for any $s^t$, $g_{\mathcal X} (s^t)= e^{ F(V(s^t),\mathcal X_{s^t}) }$. For each vector $\mathcal X$, let then $Q_t^{\mathcal X}: \mathcal B{\left(\mathbb R^{2\Delta +1} \right)}\mapsto \mathbb R$ denote the probability measure induced by $\left(V(s^t), \mathcal X (s^t)\right)$, for the assumed uniform measure on $\mathcal S_t$:
\begin{equation}
\label{def-Q-t}
Q_t^{\mathcal X} (B):= \frac{\sum_{ s^t\in \mathcal S^t} 1_{\left\{ (V,\mathcal X) \in B\right\}} (s^t)} {C_t},
\end{equation}
for arbitrary $B\in \mathcal B{\left(\mathbb R^{N^2+N}\right)}$. It is easy to verify that $Q_t^{\mathcal X}$ is indeed a probability measure. Also, we note that, for any fixed $t$ and $\mathcal X$,  $Q_t^{\mathcal X}$ is discrete, supported on the discrete set $\left\{ \left(V(s^t), \mathcal X_{s^t}\right):\,s^t\in \mathcal S_t \right\}$; note that the latter set is a subset of $\mathcal V^t\times \cup_{s^t\in \mathcal S_t} X_{s^t}$ -- the Cartesian product of the set of all feasible types at time $t$ with the set of all elements of vector $\mathcal X$.

Let $\mathbb E_{Q}$ denote the expectation with respect to measure $Q^{\mathcal X}_{t}$. Then, we have $\mathbb E_U\left[  g_{\mathcal X}(S^t)\right]=\mathbb E_Q \left[  e^{t F(V,\mathcal X)}\right]$. Going back to~\eqref{eq-expectation-1}, and using the result of Lemma~\ref{lemma-C-t-growth-rate}, we obtain for $\eta$ given in~\eqref{eq-def-eta-2}:
\begin{equation}\label{eq-expectation-2}
\eta=-\log \psi +\lim_{t\rightarrow +\infty} - \frac{1}{t} \mathbb E_0 \left[ \log \mathbb E_Q \left[  e^{t F(V,\mathcal X)}\right]\right],
\end{equation}
where, we recall $E_0$ is the expectation with respect to probability $\mathbb P_0$ that corresponds to $\mathcal H_0$ state of nature, under which measurements $X_k$ -- and hence vector $\mathcal X$ are generated.

If the measures $Q_{t}^{\mathcal X}$ were sufficiently nice such that they satisfied the LDP and the moderate growth condition~\eqref{eq-tail-condition}, then one could apply Varadhan's lemma to compute the exponential growth of the expectation in the right hand side of~\eqref{eq-expectation-2}. However, the measures $Q_{t}^{\mathcal X}$ are very difficult to analyze due to the correlations in different elements of $\mathcal X$ which couple the indicator functions in~\eqref{def-Q-t}. Hence, we resort to an upper bound of $\eta$ which we derive by replacing vector $\mathcal X$ by vector $\mathcal Z$ with the same statistical properties, but with an added feature that its elements are mutually independent. More precisely, for each $t$ we introduce a family of independent Gaussian variables $\mathcal Z=\left\{\mathcal Z_{s^t}:\,s^t\in \mathcal S^t\right\}$. Further, for each $s^t$ the corresponding element of the family $\mathcal Z_{s^t}$ is Gaussian with the same mean and variance as $\mathcal X_{s^t}$: expected value equal to $0$, and variance equal to $\mathrm {Var} \left[\mathcal Z_{s^t}\right] = \theta_2(s^t)/t$. Denote by $\mathbb P$ and $\mathbb E$, respectively, the probability function and the expectation corresponding to the family $\left\{\left\{ \mathcal Z_{s^t}:\,s^t\in \mathcal S^t \right\}   : t=1,2,\ldots\right\}$. Then, the following result holds; the proof is based on Slepian's lemma~\cite{ZeitouniNotes16}, and it can be found in Appendix.

\begin{lemma}\label{lemma-application-of-Slepian}
For each $t$, there holds,
\begin{equation}\label{eq-expectation-3}
\mathbb E \left[ \log \mathbb E_{Q} \left[  e^{t F(V,\mathcal Z)}\right]\right] \geq \mathbb E_0 \left[ \log \mathbb E_Q \left[  e^{t F(V,\mathcal X)}\right]\right],
\end{equation}
where the inner left hand side expectation is with respect to the measures $Q_t^{\mathcal X}$ and the inner right hand-side expectation is with respect to the measures $Q_t^{\mathcal Z}$.
\end{lemma}

The next result asserts that $Q_t^{\mathcal Z}$ satisfies the LDP with probability one and computes the corresponding rate function. To simplify the notation, denote $q= (1, 2, \ldots, \Delta)^\top$.
\begin{theorem}
\label{theorem-Q-t-satisfies-LDP}
For every measurable set $G$, the sequence of measures $Q_t^{\mathcal Z}$, $t=1,2,...,$ with probability one satisfies the LDP upper bound~\eqref{def-LDP-upper} and the LDP lower bound~\eqref{def-LDP-lower}, with the same rate function $I: \mathbb R^{2 \Delta+1}\mapsto \mathbb R$, equal for all sets $G$, which for $\nu \in \mathcal V$ for which $H(\nu_1)+ H(\nu_2) \geq J_{\nu}(\xi)$ is given by
\begin{equation} \label{eq-rate-function}
I(\nu,\xi)=  \log \psi - H(\nu_1) - H (\nu_2) + J_{\nu}(\xi),
\end{equation}
and equals $+\infty$ otherwise, and where, for any $\nu \in \mathcal V$, function $J_{\nu}: \mathbb  R \mapsto \mathbb R$ is defined as $J_{\nu}(\xi):=  \frac{1}{q^\top \nu_2}  \frac{\xi^2}{2}$.
\end{theorem}
The proof of Theorem~\ref{theorem-Q-t-satisfies-LDP} is given in Appendix.

Having the large deviations principle for the sequence $Q_t^{\mathcal Z}$, we can invoke Varadhan's lemma to compute the limit of the scaled values in~\eqref{eq-expectation-2}. Applying Lemma~\ref{lemma-Varadhans-wp1} (the details of the moderate growth condition~\eqref{eq-tail-condition} for $Q_t^{\mathcal Z}$ are given in Appendix, we obtain that, with probability one,
\begin{equation}\label{eq-as-limit-via-Varadhan}
\lim_{t\rightarrow +\infty} \frac{1}{t} \log \mathbb E_Q \left[  e^{t F(V,\mathcal Z)}\right]
= \sup_{(\nu,\xi)} F(\nu,\xi) - I(\nu,\xi).
\end{equation}
It can be shown that the sequence under the preceding limit is uniformly integrable; the proof of this result is very similar to the proof of a similar result in the context of hidden Markov models, given in Appendix E of~\cite{Agaskar15}, hence we omit the proof here. Thus, the limit of the sequence values and the limit of their expected values coincide, i.e.,
\begin{equation}\label{eq-UI-implies-limits-are-equal}
\lim_{t\rightarrow +\infty} \frac{1}{t} \mathbb E\left[\log \mathbb E_Q \left[  e^{t F(V,\mathcal Z)}\right]\right] =  \lim_{t\rightarrow +\infty} \frac{1}{t} \log \mathbb E_Q \left[  e^{t F(V,\mathcal Z)}\right].
\end{equation}
Combining with~\eqref{eq-expectation-2},~\eqref{eq-expectation-3}, and~\eqref{eq-as-limit-via-Varadhan}, we finally obtain
\begin{equation}\label{eq-Varadhan-final}
\eta\, \geq\, - \log \psi - \sup_{(\nu,\xi)\in \mathbb R^{2\Delta+1}} F(\nu,\xi) - I(\nu,\xi).
\end{equation}
It remains to show that the value of the above supremum equals the value of the optimization problem~\eqref{eq-main-LB-opt}. Using the definition of $I$, we have that $I(\nu,\xi)=+\infty$ for any $(\nu,\xi)$ such that $H(\nu) < J_{\theta}(\xi)$ or such that $\nu \notin \mathcal V$. Since the supremum is surely not achieved at these points,  set $\mathbb R^{2\Delta+1}$ in~\eqref{eq-Varadhan-final} can be replaced by $\left\{(\nu,\xi)\in \mathcal V\times \mathbb R: H(\nu) < J_{\theta}(\xi) \right\}$. Using the definitions of $F$ and $I$, we have
\begin{align}
&\!\!F(\nu,\xi) - I(\nu,\xi)=  \sum_{m=1}^2\sum_{d=1}^\Delta \nu_{md} \log p_{md} - \nu_{md} \log \nu_{md}\nonumber \\
&\;\;\;+  \frac{\mu_2 -\mu_1}{\sigma^2}\xi -\theta_2 \frac{\mu_2^2 - \mu_1^2}{2\sigma^2} - \frac{1}{\theta_2} \frac{\xi^2}{2 \sigma^2} - \log \psi.
\end{align}
Cancelling out the term $\log \psi$ in the preceding equation with the one in~\eqref{eq-Varadhan-final}, and recognizing that $\sum_{d=1}^\Delta \nu_{md} \log p_{md} - \nu_{md} \log \nu_{md} = - D (\nu_m || p_m)$, we see that problem~\eqref{eq-main-LB-opt} is equivalent to the one in~\eqref{eq-Varadhan-final}. This completes the proof of Theorem~\ref{theorem-main}.

\section{Numerical results}
\label{sec-NumResults}

In this section we report our numerical results to demonstrate tightness of the developed performance bounds. We also illustrate our methodology on the problem of detecting one single run of a dish-washer, where we use real-world data to estimate the state values for a dish-washer.

In the first set of simulations, we consider the setup in which $\mu_1>0$ and we compare the error exponents obtained via simulations to the guaranteed lower bound~\eqref{eq-guaranteed-exponent}. We simulate a two-state signal, $X^t$, as an i.i.d. Gaussian random variable with standard deviation $\sigma$ and mean $\mu_1=2$ and $\mu_2=5$ in states $1$ and $2$, respectively.
The duration of each state is random uniform distributed between $1$ and $\Delta=3$.
The observation interval is $t\in [1,T]$, where $T=200$. In the absence of the signal, the data is distributed according to the Gaussian distribution with mean $\mu_0=0$ and the same standard deviation $\sigma$.

To estimate the receiver operating characteristics (ROC) curves, we use $J=100000$ Monte Carlo simulation runs for each hypothesis. For each hypothesis and each simulation run, we compute the values $L_t(X^t)$, for $t=1,2,3,...,T$, using the linear recursion from Lemma~\ref{lemma-recursion-for-LRT}. Then, for each $t$, to obtain the corresponding ROC curve, we first find the minimal and maximum value $\underline L_{t,m}$ and $\overline L_{t,m}$, respectively, across $J$ runs for each hypothesis $m$, and change the detection threshold $\gamma$ with a small step size from $\underline L_{t,1}-\beta$ to $\overline L_{t,0}+\beta$, where $\beta$ is a carefully chosen bound. For each $t$ and $\gamma$ the probability of false alarm $P_{\mathrm{fa}}$ or false positive, i.e., wrongly determining that the signal is present, is calculated as
\begin{equation*}
P_{\mathrm{fa},t}^{\gamma}=\frac{\sum_{j=1}^{J}\mathbbm{1}(L_t(X^t_{(j)})\geq \gamma)}{J}
\label{eq:pfa}
\end{equation*}
where $\mathbbm{1}$ is an indicator function that returns $1$ if the corresponding condition is true and $0$ otherwise, and $X^t_{(j)}$ is the $j$-th realisation of the sequence $X^t$ under $\mathcal H_0$.
The probability of a miss $P_{\mathrm{miss}}$ or false negative, that is, declaring that the signal is not present, though it is, is calculated as:
\begin{equation*}
P_{\mathrm{miss},t}^{\gamma}=\frac{\sum_{j=1}^{J}\mathbbm{1}(L_t(X^t_{(j)})<\gamma)}{J}.
\label{eq:pmiss}
\end{equation*}
We set the bound $\alpha =0.01$ and find $P_{\mathrm{miss},t}^{\alpha}=P_{\mathrm{miss},t}^{\gamma^{\star}}$ where $\gamma^{\star}$ resulted in the highest probability of a miss that satisfied $P_{\mathrm{fa},t}^{\gamma^{\star}}\leq\alpha$.

To investigate the dependence of the slope on the SNR, we fix signal levels $\mu_1$ and $\mu_2$, and pmf's $p_1$ and $p_2$ as described above, and we vary the standard deviation of noise $\sigma$. For each different value of $\sigma$, we compute the values of $P_{\mathrm{miss},t}^{\alpha}$, for $t=1,...,T$, and apply linear regression on the sequence of values $-\log P_{\mathrm{miss},t}^{\alpha}$ for all observation times $t$ for which the probability of a miss was non-zero. This gives an estimate for the  error exponent (i.e., the slope) for the probability of a miss under a fixed value of $\sigma$, which we denote by $S_{\sigma}$.

%

Figure~\ref{fig:T2F2} plots the probability of a miss (in the logarithmic scale) vs. the number of samples $t$ for five different values of $\sigma$, namely $\sigma=10,15,20,25,30$. We observe that for large observation intervals $t$ the curves are close to linear, as predicted by the theory, see Lemma~\ref{lemma-FK-limit}. Further, as $\sigma$ increases the magnitude of the slope decreases becoming very close to $0$ for large values of $\sigma$.  Figure~\ref{fig:T2F3} compares the error exponent $S_{\sigma}$ obtained from simulations with the theoretical bound calculated using~\eqref{eq-guaranteed-exponent}. The theoretical curve is plotted in red dashed line, while the numerical curve $S_{\sigma}$ is plotted in blue full line. For comparison, we also plot the curve $\mu_2^2/(2 \sigma^2)$, which corresponds to the best possible error exponent for the studied setup, obtained when the signal throughout the whole observation interval stays at the higher signal value $\mu_2>\mu_1$; this curve is plotted in green dotted line. It can be seen from the figure that the numerical error exponent curve is at all points sandwiched between the lower bound~\eqref{eq-guaranteed-exponent} curve $\mu_1^2/(2 \sigma^2)$ and the curve $\mu_2^2/(2 \sigma^2)$. Also, the difference between the numerical error exponent and the lower bound~\eqref{eq-guaranteed-exponent} decreases as $\sigma$ increases, where the differences become negligible for large $\sigma$, showing that our bound is tight for large values of $\sigma$.

\begin{figure}[htp]
	\centering
	\includegraphics[width=1\linewidth]{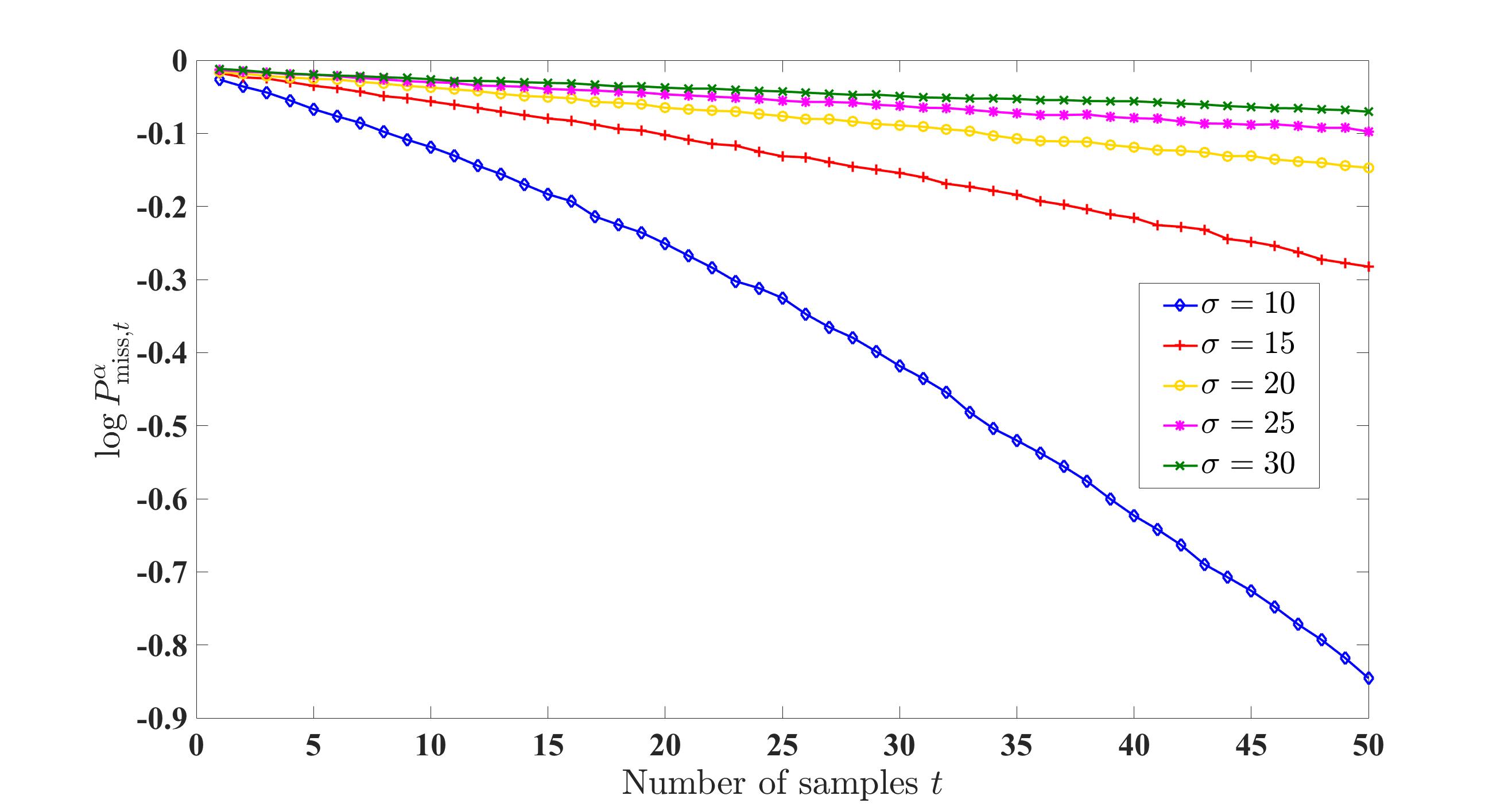}
	\caption{Simulation setup: $\Delta=3$, $p_1,p_2\sim\mathcal {U}([1,\Delta])$, $\mu_1=2$, $\mu_2=5$, $\alpha=0.01$. Evolution of probability of a miss, in the logarithmic scale, for $\sigma=5,10,15,20,25$. }
	\label{fig:T2F2}
\end{figure}

\begin{figure}[htp]
	\centering
	\includegraphics[width=1\linewidth]{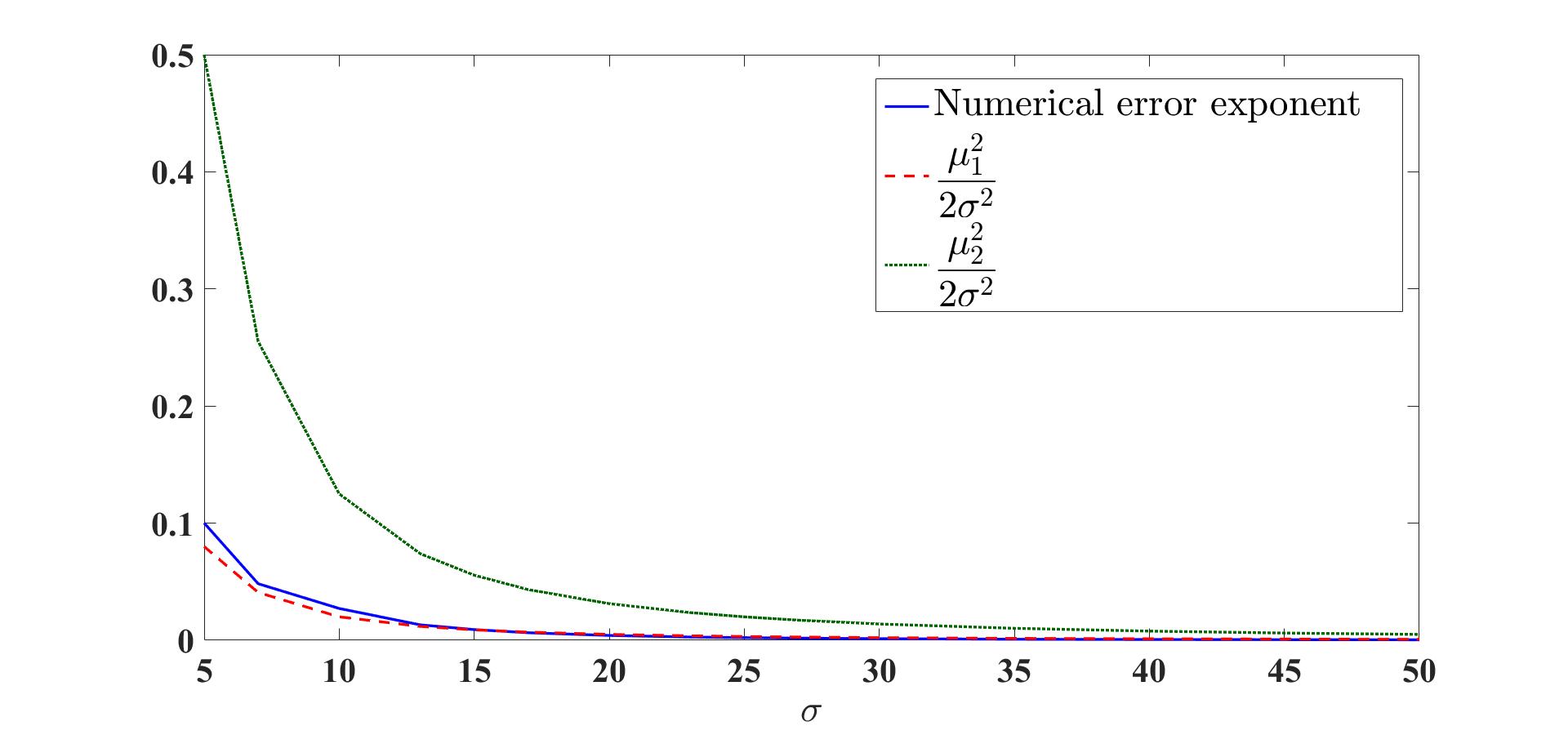}
	\caption{Simulation setup: $\Delta=3$, $p_1,p_2\sim\mathcal {U}([1,\Delta])$, $\mu_1=2$, $\mu_2=5$, $\alpha=0.01$. $\sigma$ varies from $5$ to $50$. Blue full line plots the numerical error exponent estimated from slope of $\log{P_{\mathrm{miss},t}^\alpha}$ vs. $\sigma$.  Red dashed line plots the theoretical bound $\mu_1^2/(2 \sigma^2)$ in~\eqref{eq-guaranteed-exponent}. Green dotted line plots function $\mu_2^2/(2 \sigma^2)$. }
	\label{fig:T2F3}
\end{figure}
%


In the second set of experiments, we consider the setup where the signal level in state $1$ is zero, $\mu_1=0$, and $\mu_2=\mu=1$; similarly as in the previous setup, we consider uniform distributions $p_1,p_2\sim\mathcal {U}([1,\Delta])$, with $\Delta=2$. We compare the numerical error exponent with the one obtained as a solution to optimization problem~\eqref{eq-convex-reformulation}. To solve~\eqref{eq-convex-reformulation}, we apply random search over $10^6$ different vectors from set $\mathcal V$, and pick the point which gives the smallest value of the objective (and satisfies the constraint in~\eqref{eq-convex-reformulation}).

 Figure~\ref{fig:T3F2} plots probability of a miss vs. number of samples $t$ for $5$ different values of $\sigma$, in the interval from $0.2$ to $0.6$. Again, we can observe that linearity emerges with the increase of $\sigma$. Figure~\ref{fig:T3F3}, top, compares error exponent estimated from the slope in Figure~\ref{fig:T3F2} with the theoretical bound calculated from solving~\eqref{eq-convex-reformulation}. We can see from the plot that the two lines are very close to each other. In fact, we have that the numerical values are slightly below the lower bound values. This seemingly contradictory effect is a consequence of the following. As the probability of a miss curves have a concave shape in this simulation setup (which can be observed from Figure~\ref{fig:T3F2}) their slopes continuously increase with the increase of the observation interval. As a consequence, the linear fitting performed on the whole observation interval is underestimating the slope, as it is trying to fit also the region of values where concavity is more prominent. To further investigate this effect, we performed linear fitting of probability of a miss curves only for a region of higher values of $t$, where emergence of linearity is already evident. In particular, for each different value of $\sigma$, we apply linear fitting for $[4/5\,t_{\max}, t_{\max}]$, where $t_{\max}$ is the maximal $t$ for which the probability of a miss is non-zero, and we plot the results in Figure~\ref{fig:T3F3}, bottom. It can be seen from the figure that the numerical curve got closer to the theoretical curve, indicating that the bound in~\eqref{eq-convex-reformulation} is very tight or even exact. Finally, it can be seen from Figure~\ref{fig:T3F3} (top and bottom) that the value of $\sigma$ for which the error exponent is equal to zero matches the threshold predicted by the theory, $\sigma^\star=\mu/(2 \sqrt{2\log\Delta})=0.4247$, obtained from detectability condition~\eqref{eq-detectability-condition}.

\begin{figure}[htp]
	\centering
	\includegraphics[width=1\linewidth]{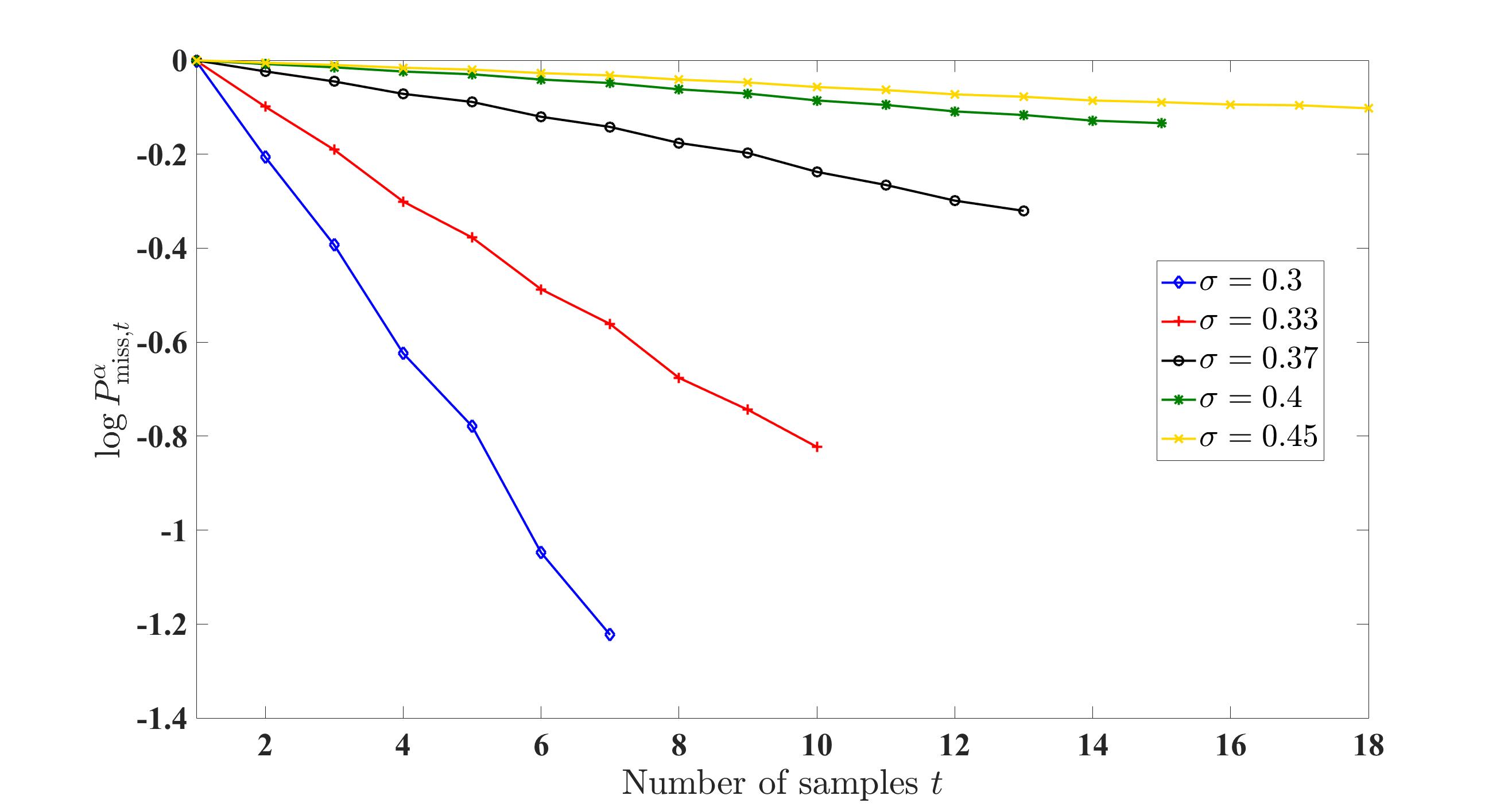}
	\caption{ Simulation setup: $\Delta=2$, $p_1,p_2\sim\mathcal {U}([1,\Delta])$, $\mu_1=0$, $\mu_2=1$, $\alpha=0.01$. Plots of probability of a miss in the logarithmic scale for $\sigma=0.3, 0.33,0.37,0.4, 0.45$}
	\label{fig:T3F2}
\end{figure}

\begin{figure}[htp]
 \centering
  \includegraphics [width=1\linewidth]{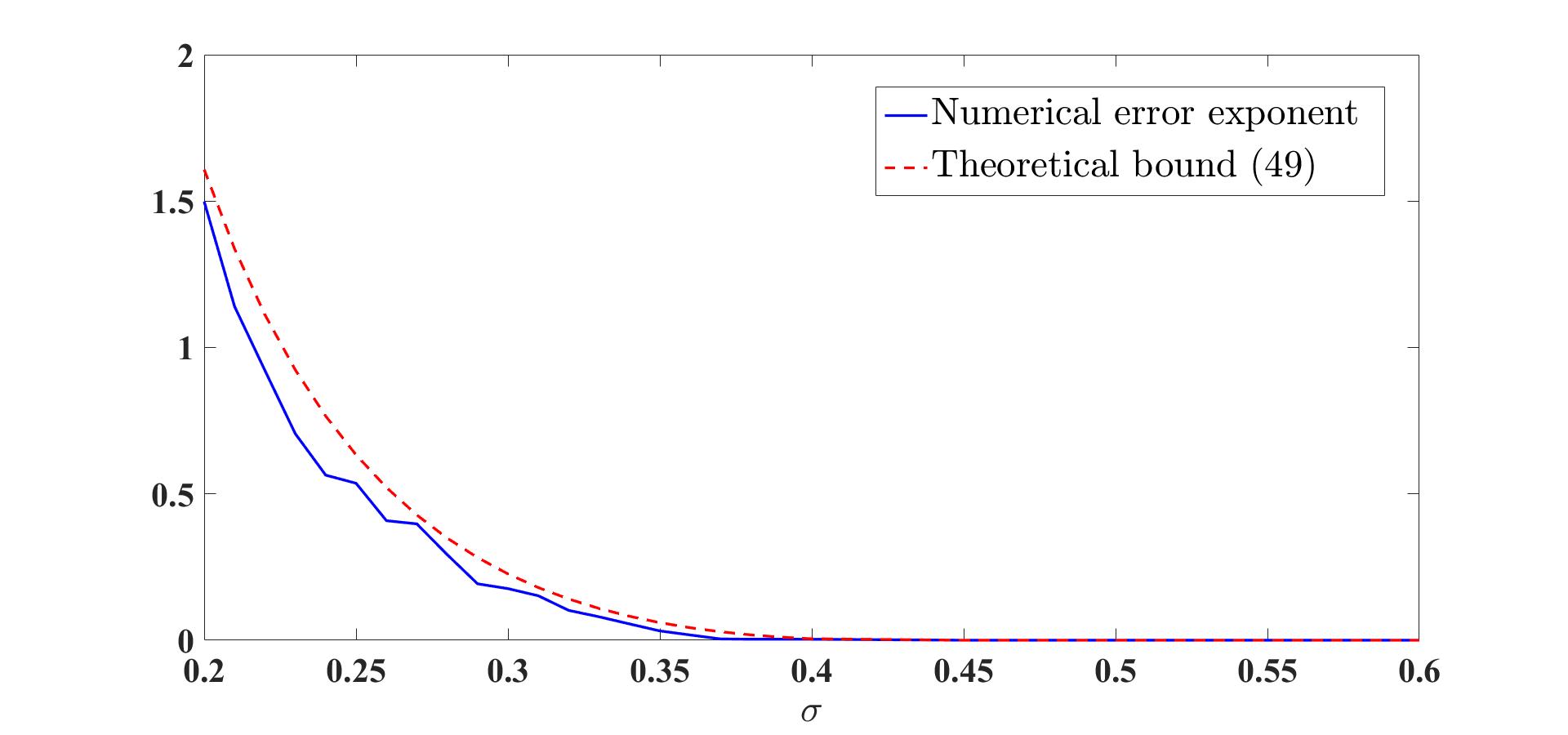}
  \includegraphics[width=1\linewidth]{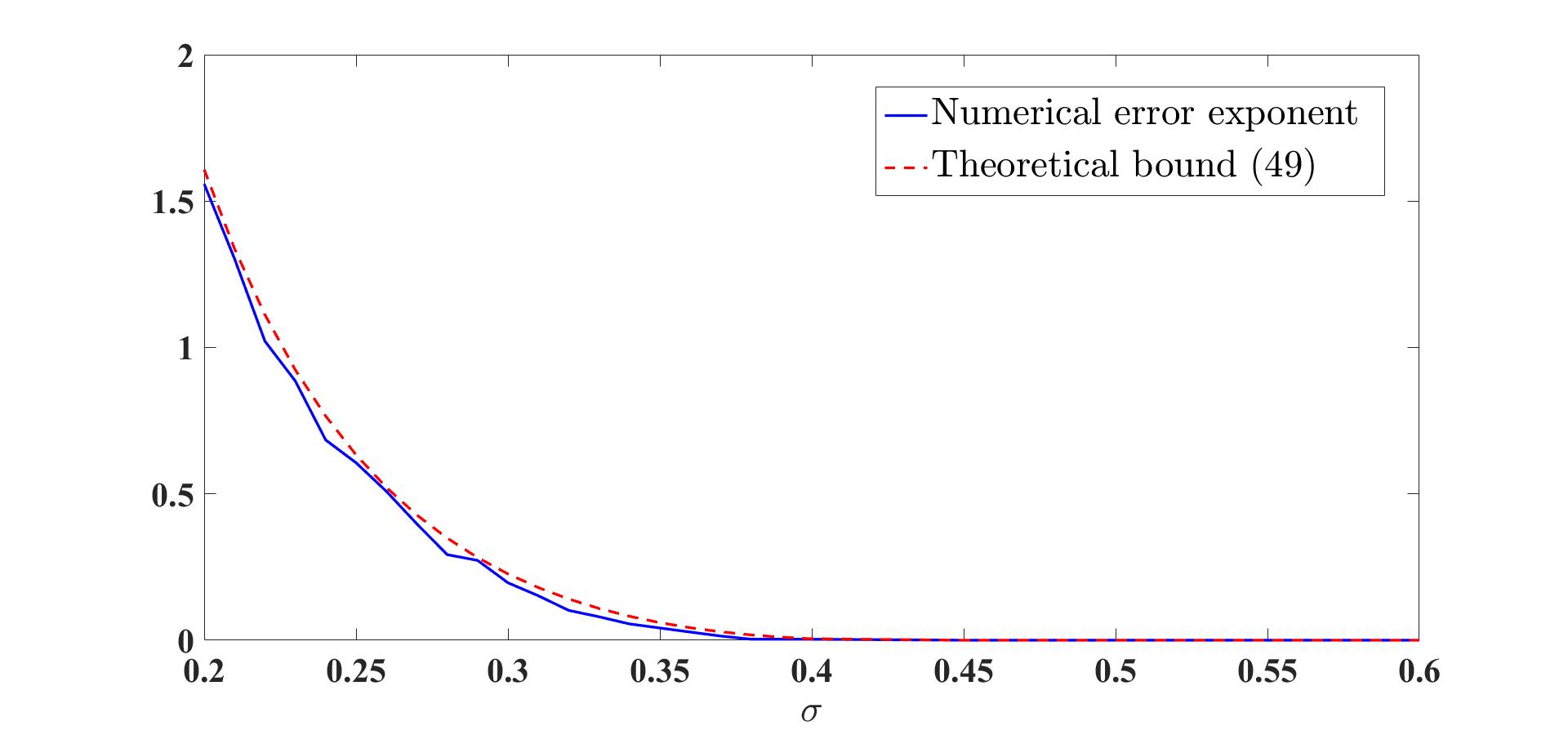}
 \caption{ Simulation setup: $\Delta=2$, $p_1,p_2\sim\mathcal {U}([1,\Delta])$, $\mu_1=0$, $\mu_2=1$, $\alpha=0.01$. $\sigma$ varies from 0.2 to 0.6. Blue full line plots the numerical error exponent estimated from slope of $\log{P_{\mathrm{miss},t}^\alpha}$ vs. $\sigma$ by linear fitting. \textbf{Top}: linear fitting performed on the whole interval $[1,t_{\max}]$; \textbf{bottom}: linear fitting performed on $[4/5\,t_{\max},t_{\max}]$.  Red dashed line plots the theoretical bound calculated by solving~\eqref{eq-convex-reformulation}). }
 \label{fig:T3F3}
\end{figure}


In the final set of simulations, we demonstrate applicability of the results to estimate the number of samples needed to detect an appliance run from the smart meter data. To do that, we use measurements of a dishwasher from the REFIT dataset \cite{Murray:2017}. REFIT dataset contains $2$ years of appliance measurements from $20$ houses. The monitored dishwasher is a two-state appliance, with mean power values of $\mu_1=2200W$, $\mu_2=66W$ and standard deviation of $\sigma_1=36.6W$ and $\sigma_2=18.2W$, in states $1$ and $2$, respectively. The mean value of background noise which is also base-load in that house is $\mu_0=90$ and with standard deviation  $\sigma_0=16.6W$. We down sampled dishwasher data with $\Delta=10$ to simulate the influence of noise including base-load and unknown appliances on detecting the appliance. 
 The simulation results are shown in Figure~\ref{fig:subfig} as plots of $P_{\mathrm{miss},t}^{\alpha}$ vs. $t$ for several values of $\sigma$ between the measured $\sigma_1$ and $\sigma_2$.


\begin{figure}
	\centering
	
	\includegraphics[width=1\linewidth]{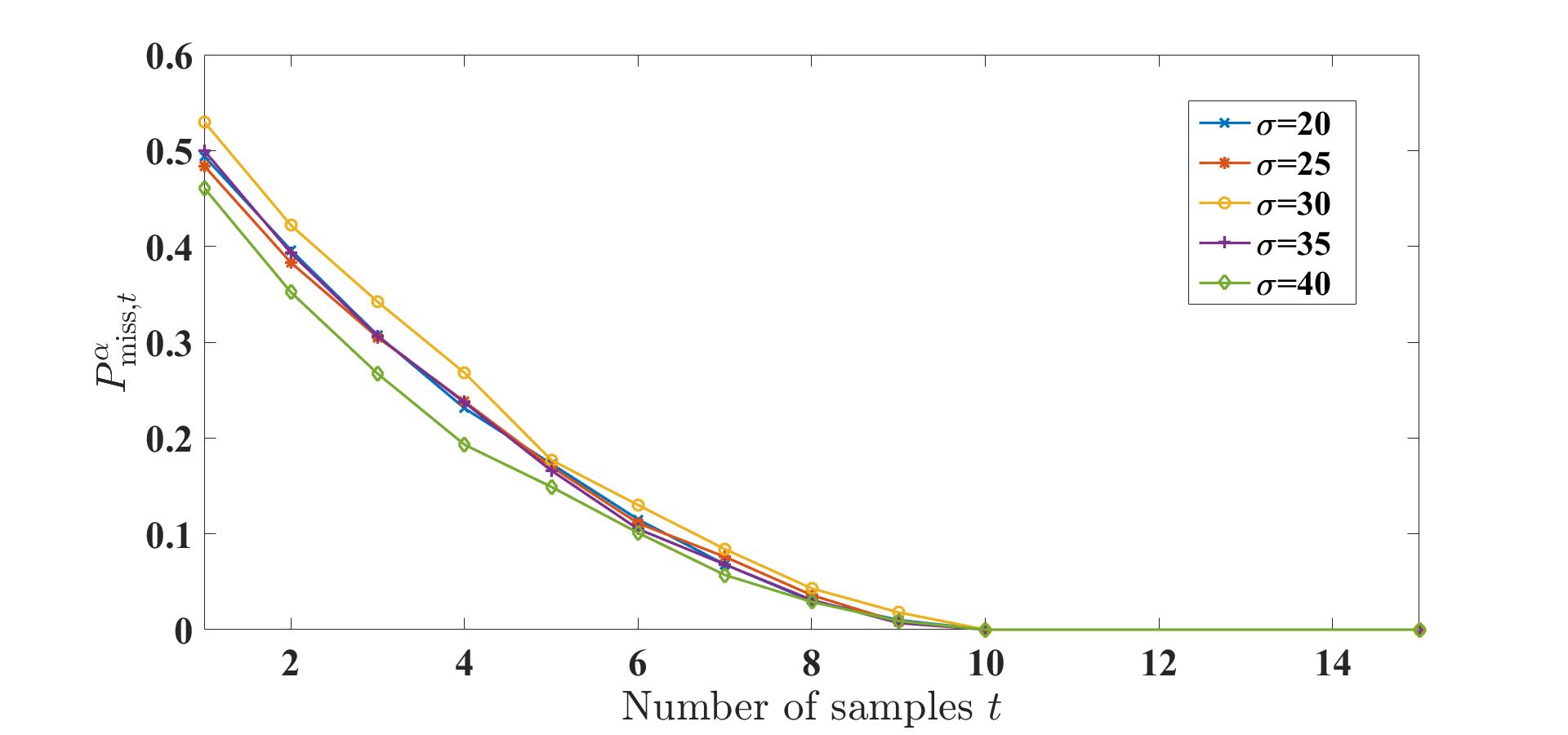}
	\caption{Simulation setup: $\Delta=10$, $p_1,p_2\sim\mathcal {U}([1,\Delta])$, $\mu_1=66$, $\mu_2=2200$, $\sigma=90$, $\alpha=0.01$. Plots of probability of a miss for 5 different $\sigma$ values. }
	
		\label{fig:subfig}
\end{figure}

As expected, the probability of a miss decreases with the increase of number of samples $t$. Furthermore, the number of samples needed for successful detection is about $10$. 

\section{Conclusion}
\label{sec-Conclusion}
We studied the problem of detecting a multi-state signal hidden in noise, where the durations of state occurrences vary over time in a nondeterministic manner. We modelled such a process via a random duration model that, for each state, assigns a (possibly distinct) probability mass function to the duration of each occurrence of that state. Assuming Gaussian noise and a process with two possible states, we derived optimal likelihood ratio test and showed that it has a form of a linear recursion of dimension equal to the sum of the duration spreads of the two states. Using this result, we showed that the Neyman-Pearson error exponent is equal to the top Lyapunov exponent for the linear recursion, the exact computation of which is a well-known hard problem. Using the theory of large deviations, we provided a lower bound on the error exponent. We demonstrated the tightness of the bound with numerical results. Finally, we illustrated the developed methodology in the context of NILM, applying it on the problem of detecting multi-state appliances from the aggregate power consumption signal.

\section*{Appendix}

\mypar{Proof of Lemma~\ref{lemma-sequence-s-t-probability}} Fix an arbitrary sequence $s^t$. Let $n_1=N_1(s^t)$, $n_2=N_2(s^t)$ and $n=N(s^t)$ denote, respectively, the number of state-$1$ phases, state-$2$ phases, and the total number of phases in $s^t$. Let the durations  of state-$1$ phases (by the order of appearance) in $s^t$ be $d_{11},d_{12},...,d_{1n_1}$, and the durations of state-$2$ phases be $d_{21},d_{22},...,d_{2n}$. Recall that ${o}(s^t)$ denotes the duration of the last phase in $s^t$.  Then, if $s_t=1$, we have
\begin{align}\label{eq-computing-prob-of-sequence}
&\!P(s^t)=\mathbb P_1(S^t=s^t) \nonumber\\
  & = \mathbb P_1\left(D_{11}=d_{11},D_{21}=d_{21},\ldots,D_{1n_1}\geq d_{1n_1} \right)\nonumber\\
  & = \prod_{l=1}^{n_1-1} \mathbb P_1\left(D_{1l}=d_{1l}\right)  \mathbb P_1 \left(D_{1n_1} \geq d_{1n_1} \right)
\prod_{l=1}^{n_2} \mathbb P_1\left(D_{2l}=d_{2l}\right),
\end{align}
where the second equality follows from the fact that the last phase is state $1$ and that with the knowledge of only up to time $t$ it is not certain whether this last phase lasts longer than $d_{1n_1}$, i.e., stretches over time $t$; the last equality follows from the fact that $D_{mn}$'s are i.i.d. for each $m$ and mutually independent for different $m$.
Adding the missing factor in the product $\mathbb P_1 \left(D_{1n_1}=d_{1n_1}\right)$, and dividing the middle term in~\eqref{eq-computing-prob-of-sequence} by the same factor, yields
\begin{align}\label{eq-computing-prob-of-sequence-2}
P(s^t) = \frac{\mathbb P_1\left(D_{1n_1} \geq d_{1n_1}\right)}{\mathbb P_1\left(D_{1n_1} = d_{1n_1}\right)} \prod_{m=1}^2\prod_{l=1}^{n_m} \mathbb P_1\left(D_{ml}=d_{ml}\right).
\end{align}
Similar formula can be obtained for the case when $s_t=2$. Note now that, for every $d=1,...,\Delta_m$, $\mathbb P_1\left(D_{ m n_m} \geq d\right) = p_{m d} + p_{m d+1}+\ldots +p_{m \Delta}=: p_{md}^{+}$, for $m=1,2$. Grouping, for each state, the product terms with equal durations, and denoting $n_{1d}=N_{1d} (s^t)$, for $d=1,...,\Delta_1$, and $n_{2d}=N_{2d} (s^t)$, for $d=1,...,\Delta_2$,  we obtain that
\begin{equation}
P(s^t) =\frac{p_{m,n_m}^{+} }{p_{m n_m}} \prod_{d=1}^{\Delta_1} p_{1d}^{n_{1d}}  \prod_{d=1}^{\Delta_2} p_{2d}^{n_{2d}}.
\end{equation}
This completes the proof of the lemma.

\mypar{Proof of Lemma~\ref{lemma-recursion-for-LRT}} Consider~\eqref{eq-LLR-2} and note that $P(s^t)$ can be expressed as $P(s^t)= p_{s_t o(s^t)}^{+}\,P^\prime(s^{t-o(s^t)})$, where, we recall ${o}: \mathcal S^t\mapsto \mathbb Z$  is an integer-valued function which returns the duration of the last phase in a sequence $s^t$. We break the sum in~\eqref{eq-LLR-2} as follows,
\begin{align}
\label{eq-LLR-3}
L_t(X^t) = \sum_{m=1}^2 \sum_{d=1}^\Delta  p_{md}^{+} \sum_{\substack{s^t \in \mathcal S^t: s_t=m,\\ {o}(s^t)=d}}  P^\prime(s^{t-d})   e^{\sum_{k=1}^t f_{s_k}(X_k) },
\end{align}
 To prove the lemma, it suffices to show that, for each $m,d,t$, the $\Lambda^m_{t,d}$'s are equal to the corresponding summands in~\eqref{eq-LLR-3},
\begin{equation}\label{eq-Sigma-t-l}
\Sigma^m_{t,d}:=\sum_{s^t \in \mathcal S^t: s_t=m, {o}(s^t)=d} P^\prime(s^{t-d})  e^{\sum_{k=1}^t f_{s_k}(X_k)}.
\end{equation}
To prove the previous claim, fix $m=1$. For $t=1$, it is easy to see that $\Sigma^1_{1,1} = e^{f_{1} (X_1)}$, and, since, when $t=1$, there cannot be sequences with last phase longer than $1$, we have $\Sigma^1_{1,d}=0$ for all $2\leq d \leq\Delta$. Analogous identities can be derived for $m=2$. Thus, we have proved that, for $t=1$, the summands $\Sigma^{m}_{t,d}=\Lambda^m_{t,d}$, for each $d$ and $m$.

Consider now an arbitrary fixed $t\geq 2$. Consider $m=1$ and $d=1$. This pair of parameter values corresponds to sequences that end with state $1$ with phase of length $1$. We thus obtain that $s_{t-1}=2$, and we can represent this set of sequences as:
\begin{align}\label{eq-sequences-singletons}
&\!\!\left\{s^{t}\in \mathcal S^t:  s_t=1, o(s^t)=1\right\} \nonumber\\
&=\bigcup_{l=1}^{\Delta}\left\{\left(s^{t-1},1\right): s^{t-1}\in \mathcal S^{t-1}, s_{t-1}=2, o(s^{t-1})=l\right\}.
\end{align}
Hence, we can write $\Sigma^1_{t,1}$ as follows:
\begin{align}\label{eq-singleton}
\Sigma^1_{t,1} & =  e^{f_1(X_t)} \sum_{l=1}^\Delta \sum_{\substack {s^{t-1} \in \mathcal S^{t-1}: s_{t-1}=2, \\o(s^{t-1})=l}}\!\!\!\!\! p_{2l} P^\prime(s^{t-1-l}) e^{\sum_{k=1}^{t-1} f_{s_k}(X_k) }\nonumber\\
& =  p_{2l} e^{f_1(X_t)} \sum_{l=1}^\Delta  \Sigma^2_{t-1,l},
\end{align}
where in the first equality we used that, when $o(s^{t-1})=l$, $P^\prime(s^{t-1}) =  p_{2l} P^\prime(s^{t-1-l})$  and the last equality follows by the definition of $\Sigma^2_{t,l}$, $l=1,2,...,\Delta$ in~\eqref{eq-Sigma-t-l}.

Consider now $m=1$ and $d\geq 2$. Since the last $d$ states must be state $1$, we can represent this set of sequences as:
\begin{align}\label{eq-sequences-l-greater-than-2}
&\left\{s^{t}\in \mathcal S^t:   s_t=s_{t-1}=\ldots =s_{t-d+1}=1, o(s^t)=d\right\} =\nonumber\\
&\left\{\left(s^{t-1},1\right): s^{t-1}\in \mathcal S^{t-1}, s_{t-1}=\ldots =s_{t-1-(d-1)+1}=1,\right.\nonumber\\
&\,\left.o(s^{t-1})=d-1\right\}.
\end{align}
Thus, we can write $\Sigma^1_{t,d}$ as follows:
\begin{align}\label{eq-l-greater-than-2}
\Sigma^1_{t,d} & = e^{f_1(X_t)} \sum_{\substack {s^{t-1} \in \mathcal S^{t-1}: s_{t-1}=1,\\ o(s^{t-1})=d-1}}\!\!\! P^\prime(s^{t-1 - (d-1)}) e^{\sum_{k=1}^{t-1} f_{s_k}(X_k) }\nonumber\\
& = e^{f_1(X_t)}\Sigma^1_{t-1,d-1},
\end{align}
where, we note that in the first equality we used that $P^\prime(s^{t-d})=P^\prime(s^{t-1-(d-1)})$.

Representing~\eqref{eq-singleton} and~\eqref{eq-l-greater-than-2} in a matrix form (we remark that derivations for $m=2$ are analogous), we recover recursion~\eqref{eq-Lambda-recursion}. Since we proved that the initial conditions are equal, i.e., $\Sigma^m_1 = \Lambda^m_1$, for $m=1,2$, we proved that $\Sigma^m_t = \Lambda^m_t$ for all $t$, which proves the claim of the lemma.

\mypar{Proof of Lemma~\ref{lemma-FK-limit}} To prove the claim, we apply Theorem~2 from~\cite{Furstenberg1960}. Note that since matrices $A_k$ are i.i.d., they are stationary and ergodic, and hence they are also metrically transitive, see, e.g.,~\cite{ShaliziNotes07}. Therefore the assumptions of the theorem are fulfilled. We now show that the condition of the theorem holds, i.e., we show that
\begin{equation}\label{eq-FK-condition}
\mathbb E_0\left[\log^{+} \|A_k\|\right] < + \infty,
\end{equation}
where $\log^{+}=\max\{\log,0\}$. It is easy to verify that $\|A_k\|\leq  e^{\max_{m=1,2} \left|f_m(X_k)\right|} C_{M_0} $, where $C_{M_0}=\|M_0\|$. Thus, we have
\begin{align}
\log^{+} \|A_k\| & \leq \log^{+} C_{M_0} e^{\max_{m=1,2} \left|f_m(X_k)\right|} \nonumber \\
&\leq \log^{+} C_{M_0} + \max_{m=1,2} \left|f_m(X_k)\right| \nonumber\\
&\leq \log^{+} C_{M_0} +  \left|f_1(X_k)\right| + \left|f_2(X_k)\right|.
\end{align}
Since $X_k$ is Gaussian, and $f_1$ and $f_2$ are linear functions, we have that $f_1(X_k)$ and $f_2(X_k)$ are Gaussian. Therefore, the expectation of the right hand side of the preceding equation is finite (which can be seen by bounding $\mathbb E_0\left[\left|f_1(X_k)\right|\right] \leq
\sqrt{\mathbb E_0\left[f_1^2(X_k)\right]} \leq +\infty $, and similarly for $m=1$). Hence, the condition~\eqref{eq-FK-condition} follows. By Theorem~2 from~\cite{Furstenberg1960} we therefore have that
\begin{equation}
\label{eq-FK-limit}
\lim_{t\rightarrow +\infty}\, \frac{1}{t}\,\log \|\Pi_t\| = \lim_{t\rightarrow +\infty} \,\frac{1}{t}\,\mathbb E\left[\log \|\Pi_t\|\right],
\end{equation}
which proves~\eqref{eq-Lyapunov-limit-expectations}.
To prove~\eqref{eq-Lyapunov-limit}, we note that $L_t = {p^{+}}^\top \Pi_t {\mathbb 1}_{2\Delta}$, where ${p^{+}}>0$. Thus, there exist constants $c$ and $C$ such that $c \|\Pi_t\| \leq L_t\leq C \|\Pi_t\|$~\cite{MatrixAnalysis}. The claim now follows from the preceding sandwich relation between $L_t$ and $\|\Pi_t\|$.

\mypar{Proof of Theorem~\ref{theorem-Q-t-satisfies-LDP}}

Fix $t\geq 1$ and fix $\nu\in \mathcal V_t$. For $D\subseteq \mathbb R$, introduce
\begin{equation}\label{eq-Q-t-nu}
Q^{\mathcal Z}_{t,\nu}(D):= \frac{\sum_{s^t\in \mathcal S^t_{\nu}} 1_{\left\{ \mathcal Z_{s^t} \in D \right\}}}{C_{t,\nu}},
\end{equation}
where, we recall, $C_{t,\nu}$ is the number of type $\nu$ feasible sequences of length $t$. Let $B =C\times D$ be a box in $\mathbb R^{2\Delta+1}$, where $C$ is a box in $\mathbb R^{2\Delta}$ and $D=[a,b]$ is an interval in $\mathbb R$. Then, we have
\begin{equation}
\label{eq-Q-t-through-Q-t-nu}
Q^{\mathcal Z}_{t}(B) = \sum_{\nu \in \mathcal V_t\cap C} \frac{C_{t,\nu}}{C_t} Q^{\mathcal Z}_{t,\nu}(D).
\end{equation}
From~\eqref{eq-Q-t-through-Q-t-nu} it follows that, for each $t$, for any $\nu \in \mathcal V_t$ there holds
\begin{equation}
\label{eq-bounds-basic}
\frac{C_{t,\nu_t}}{C_t} Q_{t,\nu}^{\mathcal Z} (D) \leq Q_t^{\mathcal Z} (B).
\end{equation}

Further, note that, for each $\nu \in \mathcal V_t$, the corresponding elements of the random vector $\mathcal Z$, $\left\{ \mathcal Z_{s^t}:\, V(s^t)=\nu\right\}$, are i.i.d., Gaussian, with mean $0$ and variance equal to $q^\top \nu_2 \frac{ \sigma^2}{t}$. Thus, $Q^{\mathcal Z}_{t,\nu}(D)$ is binomial with $C_{t,\nu}$ trials and probability of success $\mathbb E\left[Q^{\mathcal Z}_{t,\nu}(D)\right]=q_{t,\nu}(D)$ equal to
\begin{equation}\label{eq-q-t-nu}
 q_{t,\nu}(D) = \int_{a\leq x\leq b} \frac{\sqrt{t}}{\sqrt{2\pi \,q^\top \nu_2} \sigma} e^{-t \frac{x^2}{ q^\top \nu_2\sigma^2}}dx.
\end{equation}

Using the well-known bounds on the $Q$-function~\cite{Karr93}, the following bounds on $q_{t,\nu}(D)$, for an arbitrary interval $D$, are straightforward to show.
\begin{lemma}
\label{lemma-bounds-q-t-nu}
Fix $\epsilon>0$. Then, for any $D=[a,b]$, $a<b$, there holds
\begin{equation}\label{eq-q-t-nu-bounds}
 e^{-t\epsilon} e^{-t \inf_{a \leq \eta \leq b}  J_{\nu} (\eta)} \leq
q_{t,\nu} (D)\leq e^{t\epsilon} e^{-t \inf_{a \leq \eta \leq b}  J_{\nu} (\eta)}
\end{equation}
for each $\nu\in \mathcal V_t$, and all $t$ sufficiently large.
\end{lemma}

We next show that the random measures $Q_{t,\nu}^{\mathcal Z}$ approach their expected values $q_{t,\nu}$ as $t$ increases.
\begin{lemma}
\label{lemma-bounds-on-Q-t-nu}
 Fix an arbitrary $\epsilon>0$. 
\begin{enumerate}
\item
\label{lemma-bounds-Q-t-nu-upper} 
With probability one,
\begin{equation}
\label{eq-bounds-Q-t-nu-upper}
Q_{t,\nu}^{\mathcal Z} (D)\leq q_{t,\nu} (D) e^{t\epsilon},
\end{equation}
for all $\nu\in \mathcal V_t$, for all $t$ sufficiently large.
\item 
\label{lemma-bounds-Q-t-nu-lower}
Let $\nu_t\in \mathcal V_t$, $t=1,2,...$, be a sequence of types converging to $\nu^\star \in \mathcal V$. Then, with probability one, for all $t$ sufficiently large
\begin{equation}
\label{eq-bounds-Q_t-nu-lower}
Q^{\mathcal Z}_{t,\nu_t}(D) \geq q_{t,\nu^\star}(D) (1-\epsilon).
\end{equation}
\end{enumerate}
\end{lemma}

The proof of part~\ref{lemma-bounds-Q-t-nu-upper} of Lemma~\ref{lemma-bounds-on-Q-t-nu} can be obtained by considering separately the cases: 1) $\inf_{(\nu,\xi)\in B } J_{\nu}(\xi)-H(\nu)<0$ and 2) $\inf_{(\nu,\xi)\in B } J_{\nu}(\xi)-H(\nu)\geq 0$. Then, in each of the two cases the claim can be obtained by a corresponding application of Markov's inequality on a conveniently defined sequence of sets in $\Omega$. In case 1), we use $\mathcal A_t=\left\{\omega: Q_{t,\nu}^{\mathcal Z} (D)\geq q_{t,\nu}(D)e^{t\epsilon},\,\mathrm{for\,some\,}\nu\in C\cap \mathcal V_t \right\}$. Applying the union bound, together with fact that the the cardinality of $\mathcal V_t$ is polynomial in $t$ ($\left|\mathcal V_t\right|\leq (t+1)^{2\Delta}$), we obtain from condition 1) that the probabilities $\mathbb P\left(\mathcal A_t\right)$ decay exponentially with $t$. The claim in~\ref{lemma-bounds-Q-t-nu-upper} then follows by the Borel-Cantelli lemma. Similar arguments can be derived for case 2), where in the place of set $\mathcal A_t$, set $\mathcal B_t=\left\{\omega: \sum_{s^t\in \mathcal S_\nu^t} 1_{\left\{\mathcal Z_{s^t}\in D\right\} } (s^t)\geq 1,\,\mathrm{for\,some\,}\nu\in C\cap \mathcal V_t \right\}$ is used. For details we refer the reader to Section $V$-A in~\cite{RandomWalks17}.

By defining $\mathcal C_t=\left\{\omega: \left|\frac{Q_{t,\nu_t}^{\mathcal Z} (D)}{q_{t,\nu_t}(D)}-1\right|\geq \epsilon,\,\newline\mathrm{for\,some\,}\nu\in C\cap \mathcal V_t \right\}$ and applying Chebyshev's inequality, the proof of part~\ref{lemma-bounds-Q-t-nu-lower} can be derived similarly as in the proof of part~\ref{lemma-bounds-Q-t-nu-upper}. For details, see the proof of Lemma~$13$ in~\cite{RandomWalks17}.


Having the preceding technical results, we are now ready to prove the LDP for the sequence $Q_t^{\mathcal Z}$. We first prove the LDP upper bound, and then turn to the LDP lower bound.

\mypar{Proof of the LDP upper bound}
We break the proof of the LDP upper bound into the following steps. In the first step, we show that the LDP upper bound holds with probability one for all \emph{boxes} in $\mathbb R^{2\Delta+1}$. In the second step, we extend the claim to all \emph{compact sets} via the standard finite cover argument~\cite{DemboZeitouni93}. Finally, in the third step, we move from compact sets to closed sets by using the fact that $I$ has compact support. 

\emph{Step 1: LDP for boxes} Let $B=C\times D$ be an arbitrary closed box in $\mathbb R^{2\Delta+1}$, where $C$ is a box in $\mathbb R^{2\Delta}$ and $D$ is a closed interval in $\mathbb R$. To prove the LDP upper bound for box $B$, we need to show that there exists a set $\Omega^\star_1=\Omega^\star_1(B)$ which has probability one, $\mathbb P\left(\Omega^\star_1\right)=1$, such that for every $\omega\in \Omega^\star_1$, there holds
\begin{equation}
\liminf_{t\rightarrow+\infty}-\frac{1}{t}\,\log Q_t^{\mathcal Z} (B)
\leq -\overline I(B),
\end{equation}
where $\overline I(B):= \inf_{(\nu,\xi)\in B} I(\nu,\xi)$. To this end, fix $\epsilon>0$. Applying Lemma~\ref{lemma-C-t-growth-rate}, Lemma~\ref{lemma-growth-rate-C-t-nu}, Lemma~\ref{lemma-bounds-q-t-nu}, and part~\ref{lemma-bounds-Q-t-nu-upper} of Lemma~\ref{lemma-bounds-on-Q-t-nu}, together with~\eqref{eq-Q-t-through-Q-t-nu}, we have
\begin{align}\label{eq-Q-t-box-1}
Q_t^{\mathcal Z} (B) &\leq \sum_{\nu\in C\cap \mathcal V_t} e^{4t\epsilon} e^{-t \inf_{\xi \in D}J_{\nu_t} (\xi) + tH(\nu_t) -t\log \psi}\\
&\leq \left|\mathcal V_t\right| e^{4t\epsilon} e^{-t -\log \psi - t \inf_{\nu \in C\cap \mathcal V}\inf_{\xi \in D} J_{\nu} (\xi) -  H(\nu) },
\end{align}
which holds with probability one for all $t$ sufficiently large. Dividing by $t$, taking the limit $t\rightarrow +\infty$, and letting $\epsilon\rightarrow 0$, the upper bound for boxes follows.

\emph{Step 2: LDP for compact sets} The extension of the upper bound to all compact sets in $\mathbb R^{2\Delta+1}$ can be done by picking an arbitrary closed set $F$, covering it with a family of boxes $B$ of the form as in Step 1, where a ball of a conveniently chosen size is assigned to each point of $F$, and finally extracting a finite cover of $F$. As this is a standard argument in the proof of LDP upper bounds, we omit the details of the proof here and refer the reader to~\cite{DemboZeitouni93} (see, e.g., the proof of Cram\'{e}r's theorem in $\mathbb R^d$, Chapter 2.2.2 in~\cite{DemboZeitouni93}).

\emph{Step 3: LDP for closed sets} Since the rate function has compact domain, LDP upper bound for compact sets implies LDP upper bound for closed sets. This completes the proof of the upper bound.

\mypar{Proof of the LDP lower bound}
Let $U$ be an arbitrary open set in $\mathbb R^{2\Delta + 1}$. To prove the LDP lower bound we need to show that there exists a set $\Omega^\star_2=\Omega^\star_2(U)$ which has probability one, $\mathbb P\left(\Omega^\star_2\right)=1$, such that for every $\omega\in \Omega^\star_2$, there holds
\begin{equation}
\liminf_{t\rightarrow+\infty}-\frac{1}{t}\,\log Q_t^{\mathcal Z} (U)
\geq -\overline I(U).
\end{equation}
Since $I$ is non-negative at any point of its domain, it follows that $\overline I(U)$ can either be a finite non-negative number or $+\infty$. In the latter case the lower bound holds trivially, hence we focus on the case $\overline I(U)<+\infty$.

For any point $\nu \in \mathcal V$, we define a sequence of types $\nu_t\in \mathcal V_t$ converging to $\nu$, by picking, for each $t\geq 1$, an arbitrary closest neighbor of $\nu$ in the set $\mathcal V_t$\footnote{Since $\mathcal V_t$ gets denser with $t$, the sequence $\nu_t$ indeed converges to $\nu$.}, i.e.,
\begin{equation}
\label{def-theta-t-opt}
\nu_t\in\mathrm{Argmin}_{\nu \in \mathcal V_t} \left|\nu_t - \nu\right|.
\end{equation}

Now note that by the fact that $\overline I(U)$ is an infimal value, for any $\delta>0$ there must exist $(\nu,\xi)\in U$ such that $I(\nu,\xi) \leq \overline I(U)+\delta$.  If for $(\nu,\xi)$ there holds $H(\nu)- J_{\nu} (\xi) >0$, we assign $\nu^\star=\nu$ and $\xi^\star=\xi$. Otherwise, we can decrease $\xi$ in absolute value to a new point $\xi^\prime$ such that $\left(\nu,\xi^\prime\right)$ still belongs to $U$ (note that this is feasible due to the fact that $U$ is open), and for which the strict inequality $H(\nu)- J_{\nu} (\xi^\prime) >0$ holds. Assigning $\xi^\star=\xi^\prime$ we prove the existence of $(\nu^\star,\xi^\star)\in U$ such that
\begin{align}
&I(\nu^\star,\xi^\star) \leq \overline I(U)+\delta\\
&H(\nu^\star)- J_{\nu^\star} (\xi^\star) >0.
\end{align}

Let $\nu_t$ denote a sequence of points obtained from~\eqref{def-theta-t-opt} converging to $\nu^\star$. Since $U$ is open, there exists a box $B$ centered at $(\nu^\star,\xi^\star)$ that entirely belongs to $U$. This implies that there exists a closed interval $D\in \mathbb R$ such that, for sufficiently large $t$, $\nu_t\times D \subseteq U$. By the inequality in~\eqref{eq-bounds-basic}, it follows that
\begin{equation*}
Q_t^{\mathcal Z}(U) \geq Q_t^{\mathcal Z}(\{\nu_t\}\times D) = \frac{C_{t,\nu_t}}{C_t} Q_{t,\nu_t}^{\mathcal Z}(D).
\end{equation*}

Combining the lower bound on $q_{t,\nu_t}(D)$ from Lemma~\ref{lemma-bounds-q-t-nu} with part~\ref{lemma-bounds-Q-t-nu-lower} of Lemma~\ref{lemma-bounds-on-Q-t-nu}, we obtain that for sufficiently large $t$,
\begin{align*}
Q_t^{\mathcal Z}(U)&\geq q_{t,\nu_t}(D) (1-\epsilon) \frac{C_{t,\nu_t}}{C_t}\\
&\geq e^{-3 t\epsilon} e^{-t \inf_{\xi \in D} J_{\nu^\star}(\xi) + H(\nu_t) - \log \psi} (1-\epsilon).
\end{align*}
Taking the logarithm and dividing by $t$, we obtain
\begin{equation}
\label{eq-LB-almost}
\frac{1}{t}\,\log Q_t^{\mathcal Z}(U) \geq -3 \epsilon -  \inf_{\xi \in D} J_{\nu^\star}(\xi) + H(\nu_t) -
\log \psi+ \frac{\log (1-\epsilon)}{t}.
\end{equation}
As $t\rightarrow +\infty$, $\nu_t\rightarrow \nu^\star$, and by the continuity of $H$ we have that  $H(\theta_t)\rightarrow H(\theta^\star)$. Thus, taking the limit in~\eqref{eq-LB-almost} yields
\begin{align*}
\liminf_{t\rightarrow +\infty} \frac{1}{t}\,\log Q_t^{\mathcal Z}(U) & \geq -3 \epsilon -
\inf_{\xi \in D} J_{\nu^\star}(\xi)+ H(\nu^\star) - \log \psi \\
& \geq -3 \epsilon - I(\nu^\star,\xi^\star),
\end{align*}
where in the last inequality we used the fact that $\xi^\star \in D$. The latter bound holds for all $\epsilon>0$, and hence taking the supremum over all $\epsilon>0$ yields
\begin{align*}
\liminf_{t\rightarrow +\infty} \frac{1}{t}\,\log Q_t^{\mathcal Z}(U) & \geq - I(\nu^\star,\xi^\star)\\
&\geq - \inf_{(\nu,\xi) \in U} I(\nu,\xi) - \delta.
\end{align*}
Recalling that $\delta$ was chosen arbitrarily, the lower bound is proven.

\mypar{Proof of Lemma~\ref{lemma-application-of-Slepian}}
For reference, we state here the Slepian's lemma that we use in our proof.
\begin{lemma} (Slepian's lemma~\cite{ZeitouniNotes16})
\label{lemma-Slepian}
Let the function $\phi: \mathbb R^L \mapsto \mathbb R$ satisfy
\begin{equation}
\label{slepian-first-condition}
\lim_{\|x\| \rightarrow +\infty} \phi(x) e^{- \alpha \|x\|^2}=0,\mathrm{\;for\;all\;}\alpha>0.
\end{equation}
Suppose that $\phi$ has nonnegative mixed derivatives,
\begin{equation}
\label{slepian-second-condition}
\frac{\partial^2 \phi}{\partial x_l \partial x_m} \geq 0, \mathrm{\;for\;}l\neq m.
\end{equation}
Then, for any two independent zero-mean Gaussian vectors $X$ and $Z$ taking values in $\mathbb R^L$ such that $\mathbb E_X[X_l^2]= \mathbb E_Z[Z_l^2]$ and $\mathbb E_X[X_l X_m] \geq  \mathbb E_Z[Z_l Z_m]$ there holds $\mathbb E_X[\phi(X)] \geq \mathbb E_Z [\phi(Z)]$, where $\mathbb E_X$ and $\mathbb E_Z$, respectively, denote expectation operators on probability spaces on which $X$ and $Z$ are defined.
\end{lemma}

\begin{proof}
For each fixed $t$ define function $\phi_t: \mathbb{R}^{C_t} \mapsto \mathbb{R}$,
\begin{equation}
\label{def-phi-t}
\phi_t(x):= -\log \sum_{s^t\in \mathcal S^t} e^{\gamma_{s^t}\left(x_{s^t}\right)},
\end{equation}
where $x_{s^t}$ is an element of a vector $x=\left\{x_{s^t}: s^t\in \mathcal S^t\right\}\in {\mathbb {R}}^{C_t}$, whose index is $s^t$, and where each function $\gamma_{s^t}$ is defined through function $g_{\mathcal X}$, given in~\eqref{eq-expectation-1}, as $\gamma_{s^t}(x_{s^t}):=\log(g_{x}(s^t))$. Since each $\gamma_{s^t}(x_{s^t})$, $s^t\in \mathcal S^t$, grows linearly in $x$, we have that condition~\eqref{slepian-first-condition} is fulfilled.

Further, it is straightforward to show that the second partial derivative of $\phi_t$ is given by 
\begin{equation}
\label{slepian-second-condition}
\frac{\partial^2 \phi_t}{\partial x_{s^t} \partial x_{{s^t}^\prime}} = \frac{\left(\mu_2-\mu_1\right)^2}{\sigma^4} \frac{ e^{\gamma_{s^t} (x_{s^t})+ \gamma_{{s^t}^\prime} (x_{{s^t}^\prime})} }{\left(\sum_{s^t\in \mathcal S^t} e^{\gamma_{s^t} (x_{s^t})}\right)^2},
\end{equation}
which is always non-negative, and hence condition~\eqref{slepian-second-condition} is also fulfilled. 

We next verify the conditions of the lemma on the vectors $\mathcal X$ and $\mathcal Z$. Since for the same sequence $s^t$, the corresponding $\mathcal X_{s^t}$ and $\mathcal Z_{s^t}$ have the same Gaussian distribution (of mean zero and variance equal to $q^\top V_2(s^t) \sigma^2/t$, there holds  $\mathbb E_0[\mathcal X_{s^t}^2]= \mathbb E[\mathcal Z_{s^t}^2]$. Further, it is easy to see that $\mathbb E_0[\mathcal X_{s^t}\mathcal X_{{s^t}^\prime} ] = \sum_{k: s_k=s_k^{\prime}} \sigma^2\geq 0$. On the other hand, since $\mathcal Z_{s^t}$ and $\mathcal Z_{{s^t}}^\prime$ are independent for $s^t\neq {{s^t}}^\prime$, and they are both zero mean, we have $\mathbb E[\mathcal Z_{s^t}\mathcal Z_{{s^t}^\prime} ] = 0$. Therefore, the last condition of the Slepian's lemma is fulfilled. Hence, the claim of Lemma~\ref{lemma-application-of-Slepian} follows.
\end{proof}

\mypar{Proof of the moderate growth of $Q_t^{\mathcal Z}$}
\begin{proof}
Conditions that define set $\mathcal V$ imply that $1/(2\Delta)\leq q^\top \nu_2\leq 1$, $\nu\geq 0$, and hence $\mathcal V$ is compact. Further, condition $H(\nu_1)+H(\nu_2) \geq \frac{1}{q^\top \nu_2}\xi^2$, which defines the domain of the rate function $I$, implies $2\log \Delta \geq \xi^2$. Thus, $\xi$ mus be bounded in order for $I$ to be finite, which combined with the fact that $\nu$ must belong to $\mathcal V$ which is compact, shows that $I$ has compact domain. Let $B_0$ be a box that contains the domain of $I$, and let  $M_0:=\max_{(\nu,\xi) \in B_0} F(\nu,\xi)$. Since function $F$ is continuous, it must achieve maximum on $B_0$, which we denote by $M_0$. It follows that for each $M\geq M_0$,  with probability one, the integral in~\eqref{eq-tail-condition} equals zero for all $t$ sufficiently large. Thus, condition~\eqref{eq-tail-condition} is fulfilled.
\end{proof}


\bibliographystyle{IEEEtran}
\bibliography{IEEEabrv,BibliographyRandomDurationModelsNEW}
\end{document}